\documentclass[a4paper,UKenglish]{amsart}

\usepackage{babel}
\usepackage{amssymb}
\usepackage{amsthm}
\usepackage{graphicx}

\title{Complexity of Ising Polynomials}
\thanks{This work was partially supported by the Fein foundation and the graduate
school of the Technion. }
\author{Tomer Kotek}
\address{Department of Computer Science \\ Technion --- Israel Institute of Technology \\ Haifa, Israel}
\email{tkotek@cs.technion.ac.il}
\theoremstyle{plain}
\newtheorem{thm}{Theorem}
  \theoremstyle{plain}
  \newtheorem{prop}[thm]{Proposition}
  \theoremstyle{plain}
  
  \theoremstyle{plain}
  
  \theoremstyle{plain}
  \newtheorem{remarks}[thm]{Remarks}
  \theoremstyle{plain}
  \newtheorem{algorithm}[thm]{Algorithm}
  \theoremstyle{plain}
  \newtheorem{lem}[thm]{Lemma}
  \theoremstyle{definition}
  \newtheorem{defn}[thm]{Definition}
  \newtheorem{definition}[thm]{Definition}

\newif\ifskip
\skiptrue
\newenvironment{renumerate}{\begin{enumerate}}{\end{enumerate}}

\newcommand{\spP}{\mathrm{\mathbf{\#P}}}
\newcommand{\spETH}{\mathrm{\mathbf{\#ETH}}}
\newcommand{\ETH}{\mathrm{\mathbf{ETH}}}
\newcommand{\NP}{\mathrm{\mathbf{NP}}}

\newcommand{\FPT}{\mathrm{\mathbf{FPT}}}
\newcommand{\FPPT}{\mathrm{\mathbf{FPPT}}}
\DeclareMathAlphabet{\mathpzc}{OT1}{pzc}{m}{it}
\newcommand{\za}{\mathpzc{a}}
\newcommand{\zb}{\mathpzc{b}}
\newcommand{\zc}{\mathpzc{c}}
\newcommand{\zd}{\mathpzc{d}}

\newcommand{\lamb}{\lambda}
\newcommand{\vs}{\mathsf}
\renewcommand{\epsilon}{\varepsilon}
\renewcommand{\t}{\vs{t}}
\newcommand{\x}{\vs{x}}
\newcommand{\z}{\vs{z}}
\newcommand{\y}{\vs{y}}
\newcommand{\head}{\mathrm{hd}}
\newcommand{\hd}{\head}
\newcommand{\tr}{\mathrm{tr}}

\newcommand{\mc}[1]{\mathcal{#1}}
\newcommand{\mcHb}[2]{\mc{H}_{#1\setminus #2}}

\newcommand{\kt}{ {3-\tau} }

\newcommand{\Zc}{Z_{\mathrm{labeled}}}
\newcommand{\cen}{\mathrm{cent}}
\newcommand{\N}{\mathbb{N}}
\newcommand{\R}{\mathbb{R}}
\newcommand{\Q}{\mathbb{Q}}

\newcommand{\hq}{\hat{q}}
\newcommand{\Wone}{\mathrm{\mathbf{W[1]}}}
\newcommand{\tb}{{\tilde{b}}}
\newcommand{\tc}{{\tilde{c}}}
\begin{document}

\begin{abstract}
This paper deals with the partition function 
of the Ising model from statistical mechanics, which is used to study phase transitions in physical systems. 
A special case of interest is that of the Ising model with constant energies and external field. 
One may consider such an Ising system as a simple graph together with vertex and edge weights. 
When these weights are considered indeterminates, 
the partition function for the constant case is a trivariate polynomial $Z(G;\x,\y,\z)$.
This polynomial was studied with respect to its approximability by L. A. Goldberg, M. Jerrum and M. Paterson
in \cite{ar:GJP03}. 
$Z(G;\x,\y,\z)$ generalizes a bivariate polynomial $Z(G;\t,\y)$,
which was studied in by
D. Andr\'{e}n and K. Markstr\"{o}m in  
\cite{ar:AndrenMarkstrom2009}. 

We consider the complexity of $Z(G;\t,\y)$ and $Z(G;\x,\y,\z)$ in comparison to that of the Tutte polynomial,
which is well-known to be closely related to the Potts model in the absence of an external field. 
We show that $Z(G;\x,\y,\z)$ is $\spP$-hard to evaluate
at all points in $\mathbb{Q}^3$, except those in an exceptional set of low dimension,
even when restricted to simple graphs which are bipartite and planar. 
A counting version of the Exponential Time Hypothesis, $\spETH$, was introduced by
H. Dell, T. Husfeldt and M. Wahl\'{e}n in \cite{ar:DHW10} in order to study 
the complexity of the Tutte polynomial.
In analogy to their results,
we give under $\spETH$ a dichotomy theorem stating that evaluations of $Z(G;\t,\y)$ either take exponential time
in the number of vertices of $G$ to compute, or can be done in polynomial time. 
Finally, we give an algorithm for computing $Z(G;\x,\y,\z)$ in polynomial time on graphs of bounded clique-width,
which is not known in the case of the Tutte polynomial. 
\end{abstract}

\maketitle 

\section{Introduction}
\label{se:introduction}

An Ising system is a simple graph $G=(V,E)$ together with vertex and edge weights. 
Every edge $(u,v)\in E$ has an {\em interaction energy} %$J_{u,v} \in \mathbb{R}$ 
and every vertex $u\in V$ has an {\em external magnetic field strength} %$M_{u}\in \mathbb{R}$.  
associated with it. 
A function $\sigma:V\to\{\pm 1\}$ is  a {\em configuration} of the system or a {\em spin assignment}. 
The partition function of an Ising system is a generating function 
related to the probability that the system is in a certain configuration.

L. A. Goldberg, M. Jerrum and M. Paterson \cite{ar:GJP03} studied the Ising polynomial in three variables $Z(G;\x,\y,\z)$ for 
the case where both the interaction energies of an edge $(u,v)$ and the external magnetic field strength of a vertex $v$
are constant.
They consider the existence of 
fully polynomial randomized approximation schemes (FPRAS) for
the graph parameters $Z(G;\gamma,\delta,\epsilon)$, depending on the
values of $(\gamma,\delta,\epsilon)\in \mathbb{Q}^3$.
They provide approximation schemes for some regions of 
$\mathbb{Q}^3$ while showing that other regions do not 
admit such approximation schemes. % unless $\NP=\RP$. 
Approximation schemes for $Z(G;\x,\y,\z)$ were further studied in
\cite{ar:ZhangLianBai2011,arXiv:SinclairSrivastavaThurley2011}. 
M. Jerrum and A. Sinclair studied in \cite{ar:JerumSinclair93} the approximability and 
$\spP$-hardness of 
another case of the Ising model, where weights are provided as part of the input
and no external field is present.  
The bivariate Ising polynomial $Z(G;\t,\y)$, which was studied in \cite{ar:AndrenMarkstrom2009} for its combinatorial properties,
is equivalent to setting $\x=\z=\t$ in $Z(G;\x,\y,\z)$. It is shown in \cite{ar:AndrenMarkstrom2009} that $Z(G;\t,\y)$
encodes the matching polynomial, and is equivalent to a bivariate generalization of a graph polynomial introduced by 
B. L. van der Waerden in \cite{ar:Waerden1941}. 

The trivariate and bivariate Ising polynomials fall under the general framework of partition functions, 
the complexity of which has been studied extensively starting with \cite{ar:DyerGreenhill00}
and followed by \cite{ar:BulatovGrohe05,ar:GoldbergGroheJerrumThurley09,diss:Thurley2009,ar:CCL10}.
From \cite[Theorem 6.1]{diss:Thurley2009} and implicitly from \cite{ar:GoldbergGroheJerrumThurley09} 
we get that the complexity of evaluations of the Ising polynomials satisfies a dichotomy theorem, saying that
the graph parameter $Z(G;\gamma,\delta)$ is either polynomial-time computable or $\spP$-hard. 
However, $\delta$ must be positive here. 

The $q$-state Potts model deals with a similar scenario to the Ising model, except that
the spins are not restricted to $\pm 1$
but instead receive one of $q$ possible values. 
The complexity of the $q$-state Potts model 
has attracted considerable attention in the literature. 
The partition function of the Potts model 
in the case where no magnetic field is present is closely related to the Tutte polynomial $T(G;x,y)$.
It is well-known that for every 
$\gamma,\delta\in \mathbb{Q}$, except for points $(\gamma,\delta)$ in
a finite union of algebraic exceptional sets of dimension at most $1$, computing the graph parameter 
$T(G;\gamma,\delta)$ is $\spP$-hard on multigraphs, 
see \cite{ar:JVW1990}. This holds
even when restricted to bipartite planar graphs,
see \cite{ar:VertiganWelsh1992} and \cite{ar:Vertigan06}. 
In contrast, the restriction of the Tutte polynomial to the so-called {\em Ising hyperbola}, which 
corresponds to the case of Ising model with no external field, 
is tractable on planar graphs, see \cite{ar:Fisher1966,ar:Kasteleyn1967,ar:JVW1990}. 

H. Dell, T. Husfeldt and M. Wahl\'{e}n introduced in \cite{ar:DHW10} 
a counting version of the Exponential Time Hypothesis ($\spETH$), 
which roughtly states  that counting the number of satisfying assignments to a 
$3\mathrm{CNF}$ formula requires exponential time. 
This hypothesis is implied by the Exponential Time Hypothesis ($\ETH$) for decision problems
introduced by R. Impagliazzo and R. Paturi in \cite{ar:ImpagliazzoPaturi1999}. 
Under $\spETH$, the authors of \cite{ar:DHW10}  show that the computation of the Tutte polynomial 
on simple graphs 
requires exponential time in $\frac{m_G}{\log^3 m_G}$ in general, where $m_G$ in the number of edges 
of the graph. For multigraphs they show that the computation of the Tutte polynomial
generally requires exponential time in $m_G$.  
%The approximability of 
%the Tutte polynomial was studied in \cite{AFW95,pr:GoldbergJerrum07}.

In this paper we prove that the bivariate and trivariate Ising polynomials
satisfy analogs of some complexity results for the Tutte polynomial.  
For the bivariate Ising polynomial we show a dichotomy theorem stating that evaluations of $Z(G;\t,\y)$
are either $\spP$-hard or polynomial-time computable. Moreover, 
assuming the counting version of the Exponential Time Hypothesis, the bivariate Ising polynomial
require exponential time to compute. Let $n_G$ be the number of vertices of $G$. 
\begin{thm}[Dichotomy theorem for the bivariate Ising polynomial]
\label{th:mainC} 
\ \\
  For all $(\gamma,\delta)\in \mathbb{Q}^2$:
  \begin{renumerate}
   \item If $\gamma\in\{-1,0,1\}$ or $\delta=0$,
 then $Z(G;\gamma,\delta)$ is polynomial time computable.
\item Otherwise: 
\begin{itemize}
 \item $Z(G;\gamma,\delta)$ is $\spP$-hard on simple graphs, and 
 \item unless $\spETH$ fails, 
  requires exponential time in $\frac{n_G}{\log^6 n_G}$ on simple graphs.
\end{itemize}
  \end{renumerate}
 \end{thm}

We show that the evaluations of $Z(G;\x,\y,\z)$, except for those in 
a small exceptional set $B\subseteq \mathbb{Q}^3$,
are hard to compute even when restricted to simple graphs which
are both bipartite and planar. 
\begin{thm}[Hardness of the trivariate Ising polynomial]
\label{th:mainB} ~ \\
 There is a set $B\subseteq \Q^3$ such that for 
 every $(\gamma,\delta,\epsilon) \in \mathbb{Q}^3\setminus B$,
 $Z(G;\gamma,\delta,\epsilon)$ is $\spP$-hard on simple bipartite planar graphs. \\
 $B$ is a finite union of algebraic sets of dimension $2$. 
 %
 % This is by the original theorem. 
 %
 \end{thm}

Although $Z(G;\x,\y,\z)$ is hard to compute in general,
its computation on restricted classes of graphs can be tractable. 
Computing $Z(G;\x,\y,\z)$ is fixed parameter tractable with respect to tree-width 
using the general logical framework of \cite{ar:makowskyTARSKI}.
This implies in particular that $Z(G;\x,\y,\z)$ is polynomial time computable on graphs of tree-width
at most $k$, for any fixed $k$, which also follows from \cite{ar:Noble2009}. 
Likewise, the Tutte polynomial is known to be polynomial time computable on graphs of bounded tree-width, cf. \cite{ar:Andrzejak98,ar:Noble98}.
In contrast, for graphs of bounded clique-width, a width notion which generalizes tree-width, 
the best algorithm known for the Tutte polynomial is subexponential, cf. \cite{ar:GHN06}. 
We show the following:
\begin{thm}[Tractability on graphs of bounded clique-width]
\label{th:mainD} ~\\
 There exists a function $f(k)$ such that $Z(G;\x,\y,\z)$ is computable on graphs of clique-width at most $k$ 
 in running time $O\left(n_G{}^{f(k)} \right)$.
\end{thm}
In particular, $Z(G;\x,\y,\z)$ can be computed in polynomial time on graphs of 
clique-width\footnote{Rank-width can replace clique-width here and in Theorem \ref{th:mainD},
since the clique-width of a graph is bounded by a function of the rank-width of the graph.}
 at most $k$, for any fixed $k$.
On the other hand, it follows from \cite{ar:FominGolovachLokshtanov10} that, unless $\FPT=\Wone$, 
$Z(G;\x,\y,\z)$ is not {\em fixed parameter tractable with respect to clique-width},  i.e. 
  there is no algorithm for $Z(G;\x,\y,\z)$ which runs in time $O\left(q(n_G)\cdot{f(k)}\right)$
 on graphs $G$ of clique-width at most $k$ for every $k$
such that $q$ is a polynomial.

\section{Preliminaries}
\label{se:preliminaries}

\subsection{Definitions of the Ising polynomials} 
Let $G$ be a simple graph with vertex set $V(G)$ and edge set $E(G)$.
We denote $n_{G}=|V(G)|$ and $m_{G}=|E(G)|$. 
All graphs in this paper are simple and undirected unless otherwise stated. 

Given $S\subseteq V(G)$, we denote by $E_G(S)$ the set of edges in the graph induced by $S$ in $G$ 
and by $E_G(\bar{S})$ the set of edges in the graph obtained from $G$ 
by deleting the vertices of $S$ and their incident edges. 
We may omit the subscript and write e.g. $E(S)$ when the graph $G$
is clear from the context. 

\begin{defn}[The trivariate Ising polynomial]
The trivariate Ising polynomial is
\[
 Z(G;\x,\y,\z)= \sum_{S\subseteq V(G)} \x^{|E_G(S)|} \y^{|S|} \z^{|E_G(\bar{S})|} \,.
\]
\end{defn}
For every $G$, $Z(G;\x,\y,\z)$ is a polynomial in $\mathbb{Z}[\x,\y,\z]$ with positive coefficients. 
\begin{defn}[The bivariate Ising polynomial]
The bivariate Ising polynomial is obtained from $Z(G;\x,\y,\z)$ by setting $\x=\z=\t$. In other words, 
\[
 Z(G;\t,\y)= \sum_{S\subseteq V(G)} \t^{|E_G(S)|+|E_G(\bar{S})|} \y^{|S|} \,.
\]
\end{defn}
The {\em cut} $[S,\bar{S}]_{G}$
is the set of edges with one end-point in $S$ and the other in $\bar{S}=V(G)\setminus S$.
The bivariate Ising polynomial can be rewritten as follows, using that $E_G(S)$, $E_G(\bar{S})$ and $[S,\bar{S}]_G$
form a partition of $E(G)$: 
\begin{equation}
\label{eq:bivariateCutForm}
 Z(G;\t,\y)= t^{m_G} \sum_{S\subseteq V(G)} \t^{-|[S,\bar{S}]_G|} \y^{|S|} \,.
\end{equation}

The bivariate Ising polynomial is defined in this paper in a way which is slightly different from, and yet equivalent to,
the way it was defined in \cite{ar:AndrenMarkstrom2009}. 
The definition in \cite{ar:AndrenMarkstrom2009} is reminiscent of Equation (\ref{eq:bivariateCutForm}). 

In Section \ref{se:exp} we use a generalization of the bivariate Ising polynomial:
\begin{defn}
\label{defn:ZBC}
 For every $B,C\subseteq V(G)$ such that
$B\cap C = \emptyset$ we define 
\[
 Z(G;B,C;\t,\y) = \sum_{B\subseteq S \subseteq V(G) \setminus C} \t^{|E_G(S)|+|E_G(\bar{S})|} \y^{|S|}\,.
\]
\end{defn}
Clearly, $Z(G;\emptyset,\emptyset;\t,\y)= Z(G;\t,\y)$. 
In Section \ref{se:trivariate} we use a multivariate version of $Z(G;B,C;\x,\y,\z)$. 
In Section \ref{se:FPPT} we use a different multivariate generalization of $Z(G;\x,\y,\z)$. 
 
We denote by $[i]$ the set $\{1,\ldots,i\}$ for every $i\in\mathbb{N}^+$. 
\subsection{Complexity of the Ising polynomial}
Here we collect complexity results from the literature in order to 
discuss the complexity of computing, for every graph $G$, 
the trivariate (bivariate) polynomial $Z(G;\x,\y,\z)$ ($Z(G;\t,\y)$).
By {\em computing the polynomial} we mean computing the list of coefficients of monomials
$\x^i \y^j \z^k$ such that $i,k\in \{0,1,\ldots,m_G\}$ and $j\in \{0,1,\ldots,n_G\}$.

In \cite{ar:AndrenMarkstrom2009} it is shown that several graph invariants are encoded in $Z(G;\t,\y)$. 
\begin{prop}
\label{prop:complexity1}~
The following are polynomial time computable 
in the presence of an oracle to the bivariate polynomial $Z(G;\t,\y)$.
The oracle receives a graph $G$ as input and returns the matrix of coefficients of terms $\t^i\y^j$ in  $Z(G;\t,\y)$. 
\begin{itemize}
 \item 
the matching polynomial and the number of perfect matchings,
\item
the number of maximum cuts,
\end{itemize}
and, for regular graphs, 
\begin{itemize}
 \item 
the independent set polynomial and the vertex cover polynomial.
\end{itemize}
\end{prop}

The following propositions apply two hardness results from the literature to $Z(G;\t,\y)$
using Proposition \ref{prop:complexity1}. 
\begin{prop}
\label{prop:complexity2a}
$Z(G;\t,\y)$ is $\spP$-hard to compute, even when restricted to simple $3$-regular bipartite planar graphs. 
\end{prop}
\begin{proof}
 The proposition follows from a result in \cite{ar:XiaZhangZhao2007} which states
 that it is $\spP$-hard to compute $\#\mathrm{3RBP-VC}$, the number of vertex covers on input graphs restricted 
 to be $3$-regular, bipartite and planar. 
\end{proof}

For the next proposition we need the following definition which is introduced in \cite{ar:DHW10}
following \cite{ar:ImpagliazzoPaturi1999}: 
\begin{defn}[$\#$ Exponential Time Hypothesis ($\spETH$)]
Let $s$ be the infimum of the set 
\[
\{
c : \mbox{there exists an algorithm for $\#3\mathrm{SAT}$ which runs in time } O(c^{n_G}) 
\}
\]
The {\em $\#$ Exponential Time Hypothesis} is the conjecture that $s>1$. 
\end{defn}

\begin{prop}
\label{prop:complexity2b}
There exists $c>1$ such that the computation of $Z(G;\t,\y)$ requires  
$\Omega\left(c^{n_G}\right)$ time on simple graphs, unless $\spETH$ fails. 
\end{prop}
\begin{proof}
The claim follows from a result of \cite{ar:DHW10} which states that
there exists $c>1$ for which computing the number of maximum cuts in simple graphs $G$ 
takes at least $\Omega\left(c^{m_G}\right)$ time, unless the $\spETH$ fails.
It is easy to see that the problem of computing the number of maximum cuts of disconnected
graphs can be reduced to that of connected graphs and so no subexponential algorithm exists for connected graphs,
and the proposition follows since for connected graphs $n_G = O(m_G)$. 
%For (ii) the argument is similar, using an analog hardness result of \cite{ar:DHW10} for the number 
%of perfect matchings in simple bipartite graphs. 

\end{proof}
On the other hand, $Z(G;\t,\y)$ and $Z(G;\x,\y,\z)$ can be computed na\"{\i}vely in time which is exponential in $n_G$.

The three above propositions apply to $Z(G;\x,\y,\z)$ as well.

\subsection{Clique-width}
\label{se:cliquewidth}
Let $[k]=\left\{ 1,\ldots,k\right\} $. A $k$-graph is a tuple $(G,\bar{c})$
which consists of a simple graph $G$ together with labels $c_{v}\in[k]$
for every $v\in V(G)$. The class $CW(k)$ of $k$-graphs of clique-width
at most $k$ is defined inductively. Singletons belong to $CW(k)$,
and $CW(k)$ is closed under disjoint union $\sqcup$ and two other
operations, $\rho_{i\to j}$ and $\mu_{i,j}$, to be defined next.
For any $i,j\in[k]$, $\rho_{i\to j}(G,\bar{c})$ is obtained by relabeling
any vertex with label $i$ to label $j$. For any $i,j\in[k]$, $\mu_{i,j}(G,\bar{c})$
is obtained by adding all possible edges $(u,v)$ such that $c_{u}=i$
and $c_{v}=j$. The clique-width of a graph $G$ is the minimal $k$
such that there exists a labeling $\bar{c}$ for which $(G,\bar{c})$
belongs to $CW(k)$.
We denote the clique-width of $G$ by $cw(G)$.

A $k$-expression is a term $t$ which consists
of singletons, disjoint unions $\sqcup$, relabeling $\rho_{i\to j}$
and edge creations $\mu_{i,j}$, which witnesses that the graph $val(t)$
obtained by performing the operations on the singletons is of clique-width
at most $k$. Every graph of tree-width at most $k$ is of clique-width
at most $2^{k+1}+1$, cf. \cite{ar:CourcelleOlariu2000}. While
computing the clique-width of a graph is $\NP$-hard, S. Oum and P. Seymour
showed that given a graph of clique-width $k$, finding a $(2^{3k+2}-1)$-expression
is fixed parameter tractable with clique-width as parameter, cf. \cite{ar:Oum2005,ar:SeymourOum2006}.

%--------------------------------------------------------------------------------------------------------
\section{Exponential Time Lower Bound}
\label{se:exp}

In this section we prove that in general the evaluations $(\gamma,\delta)\in \mathbb{Q}^2$
of $Z(G;\t,\y)$ require exponential time to compute under $\spETH$. 
In analogy with the use of Theta graphs to deal with the complexity of the Tutte polynomial,
we define Phi graphs and use them to interpolate the indeterminate $\t$ in $Z(G;\t,\y)$. 
We interpolate $\y$ by a simple construction.

\subsection{Phi graphs}
Our goal in this subsection is to define Phi graphs $\Phi_{\mathcal{H}}$ and compute the bivariate Ising polynomial at $\y=1$
on graphs  $G\otimes \Phi_{\mathcal{H}}$ to be defined below. 
In order to define Phi graphs we must first define $L_h$-graphs. 
For every $h\in \mathbb{N}$, the graph $L_h$  is obtained from the path 
$P_{h+1}$ with $h$ edges as follows.
Let $\head(h)$ denote one of the end-points of $P_{h+1}$. Let $\tr_1(h)$ and $\tr_2(h)$ be two
new vertices. $L_h$ is obtained from $P_{h+1}$ by 
adding edges to make both $\tr_1(h)$ and $\tr_2(h)$ adjacent to all the vertices of $P_{h+1}$. 

We can also construct $L_h$ recursively from $L_{h-1}$ by 
\begin{itemize}
 \item adding a new vertex $\hd(h)$ to $L_{h-1}$,
 \item renaming $\tr_i(h-1)$ to $\tr_i(h)$ for $i=1,2$, and
 \item adding three edges to make $\hd(h)$ adjacent to $\hd(h-1)$, $\tr_1(h)$ and $\tr_2(h)$. 
 \end{itemize}
Figure \ref{fig:Figure_L56} shows $L_5$.
\begin{figure}
\caption{\label{fig:Figure_L56} The graph $L_5$ and the construction of $L_5$ from $L_4$. $L_5$ is obtained from $L_4$ by adding the vertex $\mathrm{hd}(5)$ 
and its incident edges, and renaming $\mathrm{tr}_1(4)$ and $\mathrm{tr}_2(4)$ to $\mathrm{tr}_1(5)$ and $\mathrm{tr}_2(5)$ respectively.  }
\begin{center}
\includegraphics[scale=0.7]{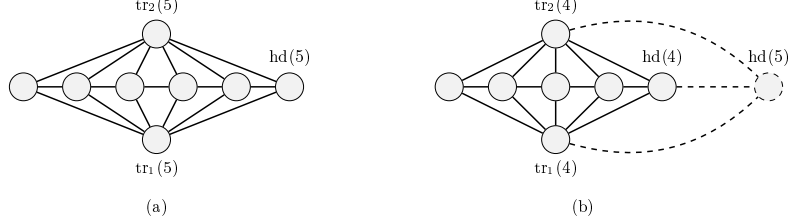}
\end{center}
\end{figure}

Let $B\sqcup C$ be a partition of the set $\{ \tr_1,\tr_2,\head\}$. 
Let $B(h)$ be the subset of $\{\tr_1(h),\tr_2(h),\head(h)\}$ which corresponds to $B$
and let $C(h)$ be defined similarly. We have that $B(h)$ and $C(h)$ form a partition
of $\{\tr_1(h),\tr_2(h),\head(h)\}$.
\begin{defn}
We denote $b_{B,C}(h) = Z(L_h;B(h),C(h);\t,1)$.
\end{defn}
The next two lemmas are devoted to computing $b_{B,C}(h)$. 
\begin{lem}
\label{lem:b_one}
 \begin{eqnarray} \notag
 b_{\{\tr_1,\head\},\{\tr_2\}}(h)& = &
 b_{\{\tr_2,\head\},\{\tr_1\}}(h)= \\ \notag
 b_{\{\tr_1\},\{\tr_2,\head\}}(h)&=&
 b_{\{\tr_2\},\{\tr_1,\head\}}(h)=
 (\t^2+\t)^h\cdot \t\,. 
 \end{eqnarray}
\end{lem}
\begin{proof}
 We have
 \begin{eqnarray} \notag
 b_{\{\tr_1,\head\},\{\tr_2\}}(h)& = &
 b_{\{\tr_2,\head\},\{\tr_1\}}(h)= \\ \notag
 b_{\{\tr_1\},\{\tr_2,\head\}}(h)&=&
 b_{\{\tr_2\},\{\tr_1,\head\}}(h)
 \end{eqnarray}
 by symmetry. 
 We compute $b_{\{\tr_1,\head\},\{\tr_2\}}(h)$ 
 by finding a simple linear recurrence relation which it satisfies and
 solving it. 
  We divide the sum $b_{\{\tr_1,\head\},\{\tr_2\}}(h)$ into two sums, 
   \begin{eqnarray*}
  b_{\{\tr_1,\head\},\{\tr_2\}}(h)  & = & Z(L_h;\{\tr_1(h),\head(h),\hd(h-1)\},\{\tr_2(h)\};\t,1) + \\
 &  & Z(L_h;\{\tr_1(h),\head(h)\},\{\tr_2(h),\hd(h-1)\};\t,1)
 \end{eqnarray*}
  depending on whether $\mathrm{hd}(h-1)$ is in 
  the iteration variable $S$ of the sum \linebreak
$b_{\{\tr_1,\head\},\{\tr_2\}}(h)$ (as in Definition \ref{defn:ZBC}). 
  These two sums can be obtained from \linebreak
 $b_{\{\tr_1,\head\},\{\tr_2\}}(h-1)$ and
 $b_{\{\tr_1\},\{\tr_2,\head\}}(h-1)$ by adjusting for the addition of $\hd(h)$
 and its incident edges:
 \begin{itemize}
  \item 
 $\hd(h-1)\in S$:  adding $\hd(h)$ (to the graph and to $S$)
 puts two new edges in $E(S)\sqcup E(\bar{S})$, namely $(\tr_1,\hd(h))$ and $(\hd(h-1),\hd(h))$. Hence, 
 \[
  Z(L_h;\{\tr_1(h),\head(h),\hd(h-1)\},\{\tr_2(h)\};\t,1)
  \]
\[
   = b_{\{\tr_1,\head\},\{\tr_2\}}(h-1) \cdot\t^2\,.
\]
 \item 
 $\hd(h-1)\notin S$: adding $\hd(h)$ puts
 just one new edge in $E(S)\sqcup E(\bar{S})$, namely $(\tr_1,\hd(h))$. Hence, 
 \[
  Z(L_h;\{\tr_1(h),\head(h)\},\{\tr_2(h),\hd(h-1)\};\t,1) 
  \]
\[
 = b_{\{\tr_1\},\{\tr_2,\head\}}(h-1)\cdot\t\,.
\]
 \end{itemize}
 Using that $b_{\{\tr_1,\head\},\{\tr_2\}}(h-1) = b_{\{\tr_1\},\{\tr_2,\head\}}(h-1)$, we get:
 \begin{gather}
 b_{\{\tr_1,\head\},\{\tr_2\}}(h) = 
 b_{\{\tr_1,\head\},\{\tr_2\}}(h-1)\cdot (\t^2+\t)
 \end{gather}
and the lemma follows since $b_{\{\tr_1,\head\},\{\tr_2\}}(0)=\t$ (note $L_0$ is simply a path of length $3$). 
\end{proof}

We are left with two distinct cases of $b_{B,C}(h)$ to compute, since by symmetry,
\begin{gather} \notag
 b_{\{\tr_1,\tr_2,\head\},\emptyset}(h) =
 b_{\emptyset,\{\tr_1,\tr_2,\head\}}(h)  
 \mbox{ and }
 b_{\{\tr_1,\tr_2\},\{\head\}}(h) =
 b_{\{\head\},\{\tr_1,\tr_2\}}(h) \,.
\end{gather}

\begin{lem}
\label{lem:b_both}
Let 
\begin{gather} \notag
 \lamb_{1,2} = \frac{\t}{2}\left(1+ \t^2 \pm\sqrt{5-2 \t^2+\t^4}\right)\,, \\ \notag
  c_1 = \t^2 - c_2 \,, \\ \notag 
  c_2 = \frac{\t\left(-\t^3-2+\t+\t\sqrt{5-2\t^2+\t^4}\right)}{2\sqrt{5-2\t^2+\t^4}} \,, \\ \notag
  d_1 = 1 - d_2 \,, \\ \notag
  d_2 = \frac{-1-2\t+\t^2+\sqrt{5-2\t^2+\t^4}}{2\sqrt{5-2\t^2+\t^4}} \,.
\end{gather}
$\lamb_1$ corresponds to the $+$ case. 
If $\t\in\mathbb{R}$ then 
$c_1,c_2,d_1,d_2,\lambda_1,\lamb_2\in \mathbb{R}$,
$\lamb_1 \not=\lamb_2$,
and 
\begin{gather} \notag
 b_{\{\tr_1,\tr_2,\head\},\emptyset}(h) = c_1 \lambda_1^h +c_2 \lamb_2^h \\ \notag
 b_{\{\tr_1,\tr_2\},\{\head\}}(h) = d_1 \lambda_1^h +d_2 \lamb_2^h\,. 
\end{gather}
 
\end{lem}
\begin{proof}
The content of the square root is always strictly positive for $\t\in \mathbb{R}$. Hence, 
$\lambda_1\not=\lambda_2$ and $c_1,c_2,d_1,d_2,\lambda_1,\lamb_2\in \mathbb{R}$. 

The sequences $b_{\{\tr_1,\tr_2,\head\},\emptyset}(h)$ and 
$b_{\{\tr_1,\tr_2\},\{\head\}}(h)$ satisfy a mutual linear recurrence as follows:
\begin{eqnarray}\notag
\label{eq:mutual_rec}
 b_{\{\tr_1,\tr_2,\head\},\emptyset}(h) &=& 
 b_{\{\tr_1,\tr_2,\head\},\emptyset}(h-1)\cdot \t^3+
 b_{\{\tr_1,\tr_2\},\{\head\}}(h-1)\cdot \t^2 \\  \notag
 b_{\{\tr_1,\tr_2\},\{\head\}}(h) &=&
 b_{\{\tr_1,\tr_2,\head\},\emptyset}(h-1) +
 b_{\{\tr_1,\tr_2\},\{\head\}}(h-1)\cdot \t 
 \end{eqnarray}
This implies that both  $b_{\{\tr_1,\tr_2,\head\},\emptyset}(h)$ and $b_{\{\tr_1,\tr_2\},\{\head\}}(h)$
satisfy linear recurrence relations
with the following initial conditions:
\begin{gather} \notag
 b_{\{\tr_1,\tr_2,\head\},\emptyset}(0) = t^2 \mbox{ and }
 b_{\{\tr_1,\tr_2,\head\},\emptyset}(1) = t^5 +t^2 
 \\ \notag
 b_{\{\tr_1,\tr_2\},\{\head\}}(0) = 1 \mbox{ and }
 b_{\{\tr_1,\tr_2\},\{\head\}}(1) = t^2 +t\,. 
\end{gather}
These recurrences can be calculated and solved using standard methods, see e.g. \cite{bk:FlajoletSedgewick2009} or \cite{bk:GrahamKnuthPatashnik94}.
\end{proof}

Using the previous two lemmas, we get:
\begin{lem}
\label{lem:Z-L_h}
\begin{eqnarray} \notag
 Z(L_h;\{\tr_1\},\{\tr_2\};\t,1) &=& Z(L_h;\{\tr_2\},\{\tr_1\};\t,1) =  2 \t (\t^2 +\t)^h \\ \notag
 Z(L_h;\{\tr_1,\tr_2\},\emptyset;\t,1) &=& Z(L_h;\emptyset,\{\tr_1,\tr_2\};\t,1) 
=
 (c_1+d_1) \lamb_1^h + (c_2+d_2) \lamb_2^h \,,
\end{eqnarray}
where $c_1,c_2,d_1,d_2,\lamb_1,\lamb_2$ are as in Lemma \ref{lem:b_both}. 
\end{lem}
\begin{proof}
 The lemma follows from Lemmas \ref{lem:b_one} and \ref{lem:b_both}.
\end{proof}

\begin{defn}[Phi graphs]
Let $\mc{H}$ be a finite set of positive integers.
We denote by $\Phi_{\mc{H}}$ the graph obtained from the
disjoint union of the graphs 
$  L_h : h\in \mc{H}$ as follows. 
For each $i=1,2$, the vertices $\tr_i(h)$, $h\in \mc{H}$, are identified into one vertex denoted $\tr_i(\mc{H})$.
\end{defn}
The number of vertices in $\Phi_{\mc{H}}$ is $2+ \sum_{h\in \mc{H}} (h+1)$.
Figure \ref{fig:Figure_H} shows $\Phi_{\{1,3,4\}}$. 
\begin{figure}
\caption{\label{fig:Figure_H} An example of a Phi graph: the graph $\Phi_{\mc{H}}$ for $\mc{H}=\{1,3,4\}$. }
\begin{center}
\includegraphics[scale=1]{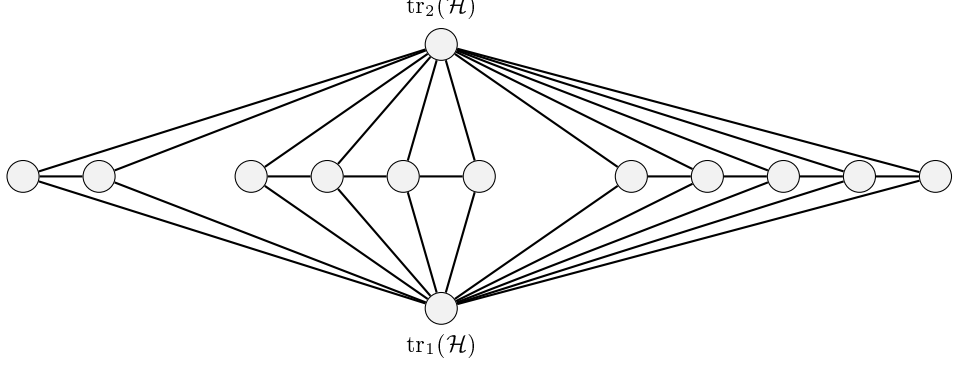}
\end{center}
\end{figure}

\begin{lem}
\label{lem:Phi_H}
 Let $\mc{H}$ be a finite set of positive integers. Then
 \begin{gather} \notag
  Z(\Phi_{\mc{H}};\{\tr_1(\mc{H})\},\{\tr_2(\mc{H})\};\t,1) = (2\t)^{|\mc{H}|} \prod_{h\in\mc{H}}(\t^2+\t)^h\,,
  \end{gather}
  and 
  \begin{eqnarray} \notag
  Z(\Phi_{\mc{H}};\{\tr_1(\mc{H}),\tr_2(\mc{H})\},\emptyset;\t,1) & = & \\ \notag 
  Z(\Phi_{\mc{H}};\emptyset,\{\tr_1(\mc{H}),\tr_2(\mc{H})\};\t,1) &=&
  \prod_{h\in \mc{H}} \left((c_1+d_1) \lamb_1^h + (c_2+d_2) \lamb_2^h\right)  \,.
 \end{eqnarray}
\end{lem}
\begin{proof}
 It follows from Lemma \ref{lem:Z-L_h} using that
 all edges are contained in some $L_h$. 
\end{proof}

We can now define the graphs $G\otimes \mc{H}$:

\begin{defn}[$G\otimes \mc{H}$]
Let $\mc{H}$ be a finite set of positive integers. 
Let $G$ be a graph. 
For every edge $e=(u_1,u_2)\in E(G)$, let  $\Phi_{\mc{H},e}$ be a new copy of $\Phi_{\mc{H}}$, where
we denote $\tr_1(\mc{H})$ and $\tr_2(\mc{H})$ for $\Phi_{\mc{H},e}$  by $\tr_1(\mc{H},e)$ and $\tr_2(\mc{H},e)$. 
Let $G\otimes \Phi_{\mc{H}} = G\otimes \mc{H}$  be the graph obtained from the disjoint union of the graphs
\[
\Phi_{\mc{H},e} : e\in E(G) 
\]
by identifying $\tr_i(\mc{H},e)$ with $u_i$, $i=1,2$, for every edge $e=(u_1,u_2)\in E(G)$. 
\footnote{It does not matter how we identify $u_1$ and $u_2$ with $\tr_1(\mc{H},e)$ and $\tr_2(\mc{H},e)$
since the two possible alignments will give raise to isomorphic graphs. }
\end{defn}

\begin{lem}
\label{lem:otimes_reduction}
Let $\mc{H}$ be a finite set of positive integers. 
Let $f_{\t,\mc{H}}$ and $g_{p,\mc{H}}$ be the following functions:
\begin{eqnarray*}
  f_{\t,\mc{H}}(e_1,e_2,r_1,r_2) &=& 
  \prod_{h\in \mc{H}}  
  \left(
  e_1 r_1^h +
  e_2 r_2^h  
  \right)\\
 f_{p,\mc{H}}(\t) &=& \left((2\t)^{|\mc{H}|} \prod_{h\in\mc{H} }(\t^2+\t)^h \right)^{m_G}\,.
\end{eqnarray*}
Then
 \begin{eqnarray} \notag
  Z(G\otimes \mc{H};\t,1) &=& f_{p,\mc{H}}(\t) \cdot  
  Z\left(G;f_{\t,\mc{H}}\left(\frac{c_1+d_1}{2t},\frac{c_2+d_2}{2t},\frac{\lamb_1}{\t^2+\t},\frac{\lamb_2}{\t^2+\t}
  \right),1\right)\,.
 \end{eqnarray}
\end{lem}
\begin{proof}
 Let $\tilde{G} = G\otimes \mc{H}$. 
 By definition,
 \[
  Z(\tilde{G};\t,1) =  \sum_{S\subseteq V(\tilde{G})} \t^{|E_{\tilde{G}}(S)\sqcup E_{\tilde{G}}(\bar{S})|}\,.
 \]
We can rewrite this sum as
\begin{eqnarray} \notag
  Z(\tilde{G};\t,1) &=&
 \sum_{S\subseteq V(G)} 
  \left(\prod_{e\in [S,\bar{S}]_G}
  Z(\Phi_{\mc{H},e};\{\tr_1(\mc{H},e)\},\{\tr_2(\mc{H},e)\}; \t,1)\right) \\ \notag
  & & \cdot \left(\prod_{e\in E_{G}(S)\sqcup E_{G}(\bar{S})}
  Z(\Phi_{\mc{H},e};\{\tr_1(\mc{H},e),\tr_2(\mc{H},e)\},\emptyset; \t,1)\right) 
  \,,
\end{eqnarray}
 since edges only occur within some $\Phi_{\mc{H},e}$. 
 Using lemma \ref{lem:Phi_H},  the sum in the last equation can be written as
  \begin{eqnarray} \notag
  & & \sum_{S\subseteq V(G)} 
  \left((2\t)^{|\mc{H}|} \prod_{h\in\mc{H} }(\t^2+\t)^h \right)^{|[S,\bar{S}]_G|} \cdot \\ & & \notag
  \left(
  \prod_{h\in \mc{H}} \left((c_1+d_1) \lamb_1^h + (c_2+d_2) \lamb_2^h\right)
  \right)^{|E_G(S)\sqcup E_G(\bar{S})|}\,.
  \end{eqnarray}
Since $|[S,\bar{S}]_G = m_G - |E_G(S)\sqcup E_G(\bar{S})|$, we can rewrite the last equation as
\begin{eqnarray} \notag & &
  \left((2\t)^{|\mc{H}|} \prod_{h\in\mc{H} }(\t^2+\t)^h \right)^{m_G} \cdot
   \\ \notag & & 
  \sum_{S\subseteq V(G)}
  \left(
  \frac{\prod_{h\in \mc{H}} \left((c_1+d_1) \lamb_1^h + (c_2+d_2) \lamb_2^h\right)}
  {(2\t)^{|\mc{H}|} \prod_{h\in\mc{H} }(\t^2+\t)^h}
  \right)^{|E_G(S)\sqcup E_G(\bar{S})|}\,.
\end{eqnarray}
The last sum can be rewritten as
\[
 \sum_{S\subseteq V(G)} 
  \left[
  \prod_{h\in \mc{H}}  
  \left(
  \frac{c_1+d_1}{2t} \left(\frac{\lamb_1}{\t^2+\t}\right)^h +
  \frac{c_2+d_2}{2t} \left(\frac{\lamb_2}{\t^2+\t}\right)^h  
  \right)
  \right]^{|E_G(S)\sqcup E_G(\bar{S})|}
\]
and the lemma follows.
\end{proof}

The construction described above will be useful to deal with 
the evaluation of $Z(G;\t,\y)$ with $\y=-1$ 
due to the following lemma.
For a graph $G$, let $G_{(1)}$ be the graph obtained from
$G$ by adding, for each $v\in V(G)$, a new vertex $v'$
and an edge $(v,v')$. So $v'$ is adjacent to $v$ only. 
$G_{(1)}$ is a graph with $2 n_G$ vertices.
\begin{lem}\label{lem:Zminus1}
 $Z(G;\t,1)=(\t-1)^{-n_G} Z(G_{(1)};\t,-1)$
\end{lem}
\begin{proof}
By definition we have
\begin{eqnarray*}
Z(G_{(1)};\t,-1) & = & \sum_{S\subseteq V(G_{(1)})} \t^{|E_{G_{(1)}}(S) \sqcup E_{G_{(1)}}(\bar{S})|} (-1)^{|S|} \\
& = & \sum_{S\subseteq V(G)} \t^{|E_G(S) \sqcup E_G(\bar{S})|} (\t - 1)^{|S|} (\t - 1)^{n_G - |S|} 
\end{eqnarray*}
where the last equality is by considering the contribution of $v'$ for each $v\in V(G)$: if $v\in S$
then $v'$ contributes either $-\t$ or $1$; if $v\notin S$ then $v'$ contributes either $\t$ or $-1$. 
The last expression in the equation above equals
\begin{eqnarray*}
(\t-1)^{n_G} \sum_{S\subseteq V(G)} \t^{|E_G(S) \sqcup E_G(\bar{S})|} & = & (\t-1)^{n_G}\cdot Z(G;\t,1)\,.
\end{eqnarray*}
\end{proof}

\subsection{The Ising polynomials of certain trees}

We denote be $S_n$ the star with $n$ leaves. Let $\cen(S_n)$ be the central vertex of the star $S_n$. 
A construction based on stars will be used to interpolate the $\y$ indeterminate from 
$Z(G;\gamma,\delta)$. First, notice the following:
\begin{prop}
\label{prop:star_ising}
For every $n\in \N^+$, 
\begin{eqnarray} \notag
 Z(S_n;\{\cen(S_n)\},\emptyset;\t,\y) &=& \y\cdot (\y\t+1)^n\\ \notag
 Z(S_n;\emptyset, \{\cen(S_n)\};\t,\y) &=& (\y+\t)^n
\end{eqnarray}
\end{prop}
\begin{proof}
By definition, 
\begin{eqnarray} \notag
 Z(S_n;\{\cen{S_n}\},\emptyset;\t,\y) &=& \sum_{S: \{\cen(S_n)\}\subseteq S\subseteq V(S_n)} \t^{|E_{S_n}(S)\sqcup E_{S_n}(\bar{S})|} \y^{|S|}\\ \notag
 Z(S_n;\emptyset,\{\cen{S_n}\};\t,\y) &=& \sum_{S\subseteq V(S_n)\setminus \{\cen(S_n)\}} \t^{|E_{S_n}(S)\sqcup E_{S_n}(\bar{S})|} \y^{|S|}
\end{eqnarray}
Consider a leaf $v$  of $S_n$.
For $Z(S_n;\{\cen(S_n)\},\emptyset;\t,\y)$, a leaf $v$ has two options: either $v\in S$, in which case
it contributes the weight of its incident edge, so its contribution is $\y\t$; or $v\not\in S$, in which case it 
contributes $1$. 
For $Z(S_n;\emptyset, \{\cen(S_n)\};\t,\y)$, $v$ has two options: either $v\in S$, in which case
it does not contribute the weight of its edge, so its contribution is $\y$; or $v\not\in S$, in which case
its edge contributes $\t$. 
\end{proof}

\begin{definition}[The graph $S_\mc{H}$]
Let $\mc{H}$ be a set of positive integers. 
The graph $S_\mc{H}$ is obtained from the disjoint union of $S_n : n\in \mc{H}$
and a new vertex $\cen(\mc{H})$ by adding edges  between 
$\cen(\mc{H})$ and the centers $\cen(S_n)$ of all the stars $S_n : n\in \mc{H}$. 
\end{definition}
See Figure \ref{fig:Figure_S_H}(a) for an example. 
\begin{figure}
\caption{\label{fig:Figure_S_H} Examples of $S_{\mc{H}}$ and $S_\mc{H}(G)$. In (b), the black vertices $u,v$ are the vertices
of $K_2$. They are also denoted $\cen(\mc{H},u)$ and $\cen(\mc{H},v)$ respectively. }
\begin{center}
\includegraphics[scale=1]{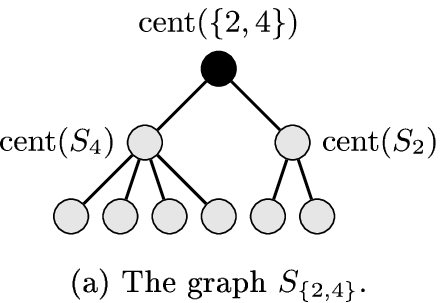}  \ \ \  \ \ \ \ \  \ \ 
\includegraphics[scale=1]{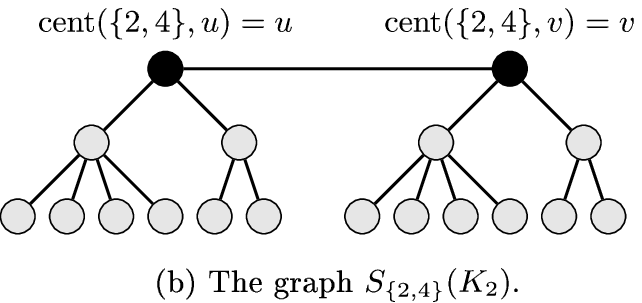}
\end{center}
\end{figure}

\begin{prop} \label{prop:star_mcH_ising}
 Let $\mc{H}$ be a set of positive integers. Then, 
\begin{eqnarray} \notag
 Z(S_\mc{H};\{\cen(\mc{H})\},\emptyset;\t,\y) &= \y \cdot \displaystyle{\prod_{h\in \mc{H}}} & \big( \y \t\cdot (\y \t +1)^h + (\y+\t)^h\big) \\ \notag
 Z(S_\mc{H};\emptyset,\{\cen(\mc{H})\};\t,\y) &= \displaystyle{\prod_{h\in \mc{H}}} & \big(\y \cdot(\y\t+1)^h +\t \cdot (\y+\t)^h \big)
\end{eqnarray}
\end{prop}
\begin{proof}
%First, 
%\begin{eqnarray}\notag
%  Z(S_\mc{H};\t,\y) &=& \y \cdot Z(S_\mc{H};\{\cen(\mc{H})\},\emptyset;\t,\y)\\ \notag
%& & \ + Z(S_\mc{H};\emptyset,\{\cen(\mc{H})\};\t,\y)\,,
%\end{eqnarray}
%and we have
We have
\begin{eqnarray} \notag
 Z(S_\mc{H};\{\cen(\mc{H})\},\emptyset;\t,\y) &= \y\cdot \displaystyle{\prod_{h\in \mc{H}}} & \big( \t\cdot  Z(S_h;\{\cen(S_h)\},\emptyset;\t,\y) \\ \notag
& & \ + Z(S_h;\emptyset,\{\cen(S_h)\};\t,\y) \big) \\ \notag  & &   \\  \notag
Z(S_\mc{H};\emptyset,\{\cen(\mc{H})\};\t,\y) &= \displaystyle{\prod_{h\in \mc{H}}} & \big( Z(S_h;\{\cen(S_h)\},\emptyset;\t,\y) \\ \notag 
& & \ + \t \cdot Z(S_h;\emptyset,\{\cen(S_h)\};\t,\y)\big)
\end{eqnarray}
and by Proposition \ref{prop:star_ising}, the claim follows. 
\end{proof}

\begin{definition}[The graph $S_\mc{H}(G)$]
 Let $\mc{H}$ be a set of positive integers and let $G$ be a graph. 
 For every vertex $v$ of $G$, let $S_{\mc{H},v}(G)$ be a new copy of $S_\mc{H}$. 
 We denote the center of each such copy of $S_\mc{H}$ by $\cen(\mc{H},v)$. 
 Let $S_\mc{H}(G)$ be the graph obtained from the disjoint union of the graphs
 in the set 
\[
 \{G\}\cup \{S_{\mc{H},v} : v\in V(G)\}
\]
 by identifying the pairs of vertices $v$ and $\cen(\mc{H},v)$. 
\end{definition}
In other words, $S_{\mc{H}}(G)$ is the {\em rooted product} of $G$ and $\left(S_{\mc{H}},\cen(\mc{H})\right)$. 
See Figure \ref{fig:Figure_S_H}(b) for an example. 

\begin{prop} \label{prop:S_big_ising}
Let $\mc{H}$ be a set of positive integers. 
Let 
\begin{eqnarray*}
      g_{p,\mc{H}}(\t,\y)&=& \left(\displaystyle{\prod_{h\in \mc{H}}}  \big( \y \cdot(\y\t+1)^h +\t \cdot (\y+\t)^h\big)\right)^{|V(G)}\\
 g_{\y,\mc{H}}(\t,\y) &=& \y \displaystyle{\prod_{h\in \mc{H}}} \frac{\y \t\cdot (\y \t +1)^h + (\y+\t)^h}{\y \cdot(\y\t+1)^h +\t \cdot (\y+\t)^h}
\end{eqnarray*}
Then
\[
 Z(S_{\mc{H}}(G);\t,\y) = g_{p,\mc{H}}(\t,\y) \cdot Z(G;\t,g_{\y,\mc{H}}(\t,\y))\,. 
\] 
\end{prop}
\begin{proof}
By definition
\[
 Z(S_{\mc{H}}(G);\t,\y) = \sum_{S\subseteq V(S_{\mc{H}}(G))} \t^{|E_{S_{\mc{H}}(G)}(S)\sqcup E_{S_{\mc{H}}(G)}(\bar{S})|} \y^{|S|}\,. 
\]
We would like to rewrite this sum as a sum over $S\subseteq V(G)$. By the structure of $S_{\mc{H}}(G)$, 
\begin{eqnarray} \notag
Z(S_{\mc{H}}(G);\t,\y) &=&  \sum_{S\subseteq V(G)} \t^{|E_{G}(S)\sqcup E_{G}(\bar{S})|} \\ \notag
& & \left(\prod_{v\in S}  Z(S_{\mc{H},v};\{\cen(\mc{H},v)\},\emptyset;\t,\y) + \right. \\ \notag & & 
\ \ \left. \prod_{v\in \bar{S}}  Z(S_{\mc{H},v};\emptyset,\{\cen(\mc{H},v)\};\t,\y)\right)
\end{eqnarray}
By Proposition \ref{prop:star_mcH_ising}, 
\begin{eqnarray} \notag
Z(S_{\mc{H}}(G);\t,\y) &=&  \sum_{S\subseteq V(G)} \t^{|E_{G}(S)\sqcup E_{G}(\bar{S})|} \\ \notag
& &  \left(
\left(\y\cdot \displaystyle{\prod_{h\in \mc{H}}}  \big( \y \t\cdot (\y \t +1)^h + (\y+\t)^h\big)\right)^{|S|}\right. \\ \notag
& & \ \ \left. \left(\displaystyle{\prod_{h\in \mc{H}}}  \big( \y \cdot(\y\t+1)^h +\t \cdot (\y+\t)^h\big)\right)^{|V(G)\setminus |S|}  \right)
\end{eqnarray}
and the claim follows. 
 
\end{proof}

\begin{samepage}
The following propositions will be useful:
\begin{prop} \label{prop:S_big2_ising}
Let $g_{\y,\mc{H}}(\t,\y)$ be as in Proposition \ref{prop:S_big_ising}. 
Let $h_{\y,\mc{H}}$ be the function given by 
\[
h_{\y,\mc{H}}(e_1,e_2,r) = \displaystyle{\prod_{h\in \mc{H}}}
\left(1+
\frac{1}
{e_1 
 +e_2 \cdot r^h} \right)\,.
\]
Let $\gamma,\delta\notin \{-1,0,1\}$ such that $\gamma\not=-\delta$. 
There exist constants $h_1,u_1,u_2,w$ (which depend on $\gamma$ and $\delta$) such that
for every two finite sets of positive even numbers $\mc{H}_1$ and $\mc{H}_2$ which satisfy
\begin{itemize}
 \item $|\mc{H}_1|=|\mc{H}_2|$, and  $\mc{H}_1,\mc{H}_2 \subseteq \N^+ \setminus\{1,\ldots,h_1\}$,
\end{itemize}
we have
\begin{renumerate}
 \item $g_{\y,\mc{H}_1}(\gamma,\delta),g_{\y,\mc{H}_1}(\gamma,\delta),h_{\y,\mc{H}_1}(u_1,u_2,w),h_{\y,\mc{H}_2}(u_1,u_2,w)  \in \R\setminus\{0\}$, and
 \item 
$g_{\y,\mc{H}_1}(\gamma,\delta) = g_{\y,\mc{H}_2}(\gamma,\delta)$ iff
$h_{\y,\mc{H}_1}(u_1,u_2,w) = h_{\y,\mc{H}_2}(u_1,u_2,w)$
\end{renumerate}
Furthermore, $u_1$ and $u_2$ are non-zero and $w\notin\{-1,0,1\}$. 
\end{prop}
\end{samepage}
\begin{proof}
It cannot hold that $|\delta\gamma+1|=|\delta+\gamma|$. Furthermore we know that $\gamma,\delta\not =0$. Hence, there is $h_1$ such that for every even
$h>h_1$,the sequences 
$\delta (\delta \gamma+1)^h +\gamma \cdot (\delta+\gamma)^h$ 
and $\delta\gamma(\delta\gamma+1)^h+(\delta+\gamma)^h$
are strictly ascending or descending, and in particular,
are non-zero. Therefore we have 
$g_{\y,\mc{H}_1}(\gamma,\delta),g_{\y,\mc{H}_1}(\gamma,\delta)\in \R \setminus \{0\}$.

We have for $i=1,2$
\begin{eqnarray}\notag
g_{\y,\mc{H}_i}(\gamma,\delta ) &= & 
 \delta \displaystyle{\prod_{h\in \mc{H}_i}}
 \frac{\delta \gamma\cdot (\delta \gamma +1)^h + (\delta+\gamma)^h}{\delta (\delta \gamma+1)^h +\gamma \cdot (\delta+\gamma)^h} \\  \notag
& = &
\frac{\delta}{\gamma^{|\mc{H}_i|}} \displaystyle{\prod_{h\in \mc{H}_i}}
 \frac{\delta \gamma^2\cdot (\delta \gamma +1)^h + \gamma \cdot(\delta+\gamma)^h}{\delta (\delta \gamma+1)^h +\gamma \cdot (\delta+\gamma)^h} \\ \notag 
& = &
\frac{\delta}{\gamma^{|\mc{H}_i|}} \displaystyle{\prod_{h\in \mc{H}_i}}\left(1+
 \frac{\delta (\gamma^2 -1)\cdot (\delta \gamma +1)^h }{\delta (\delta \gamma+1)^h +\gamma \cdot (\delta+\gamma)^h} \right)\\ \notag
& = & \frac{\delta}{\gamma^{|\mc{H}_i|}} \displaystyle{\prod_{h\in \mc{H}_i}}
\left(1+
\frac{1}
{\frac{1}{\gamma^2-1} 
 +\frac{\gamma}{\delta(\gamma^2-1)} \cdot \left(\frac{\delta+\gamma}
{\delta \gamma +1}\right)^h} \right)\\ \notag
\end{eqnarray}
Let $u_1 = \frac{1}{\gamma^2-1}$, $u_2 = \frac{\gamma}{\delta(\gamma^2-1)}$ and $w =  \frac{\delta+\gamma}
{\delta \gamma +1}$. We have $u_1,u_2,w\in \R\setminus \{0\}$ and $w\notin \{-1,1\}$. 
Hence, we can take $h_1$ to be large enough so that $u_1+u_2+w^h$ non-zero. Since 
$u_1+u_2+w^h$ is strictly ascending or descending for even $h$, we have 
$h_{\y,\mc{H}_1}(u_1,u_2,w),h_{\y,\mc{H}_2}(u_1,u_2,w)  \in \R\setminus\{0\}$
for large enough values of $h$. 
\end{proof}

\begin{prop} \label{prop:h_2_ising}
Let  $\gamma,\delta\notin\{-1,0,1\}$ and $\gamma\not=-\delta$. 
Let $\mc{H}$ be a set of positive even integers. 
Let $g_{p,\mc{H}}(\t,\y))$ be from Proposition \ref{prop:S_big_ising}. 
Then there exists $h_2$ such that if $\mc{H} \subseteq \N^+ \setminus \{1,\ldots,h_2\}$ 
then $g_{p,\mc{H}}(\gamma,\delta))\not=0$. 
\end{prop}
\begin{proof}
Recall
\[
 g_{p,\mc{H}}(\gamma,\delta))=\left(\displaystyle{\prod_{h\in \mc{H}}}  \big( \delta (\delta\gamma+1)^h +\gamma \cdot (\delta+\gamma)^h\big)\right)^{|V(G)|}\,.
\]
We have that $\delta+\gamma$ is non-zero.
If $\delta\gamma+1=0$ then the claim holds even for $h_2=0$. 
Otherwise, using that $|\delta\gamma+1|\not=|\delta+\gamma|$,
 at least one of $(\delta\gamma+1)^h,(\delta+\gamma)^h$ becomes strictly larger in absolute value than the other for large enough $h$. 
\end{proof}

%-------------------------------------------------------------------------------------------------------------------------
%-------------------------------------------------------------------------------------------------------------------------
%-------------------------------------------------------------------------------------------------------------------------
%-------------------------------------------------------------------------------------------------------------------------
\subsection{Proof of Theorem \ref{th:mainC}}

The following lemma is a variation of  Lemma 4 in \cite{ar:DHW10}.

For any $\mc{H}$, let $\sigma(\mc{H})= {\sum_{h\in \mc{H}} h}$. 
\begin{lem}
\label{lem:dell-lemma}~\\
Let $\gamma\notin\{-1,0,1\}$, $\delta\not=0$, $e_1,e_2\not=0$ and $r_1,r_2\notin \{-1,0,1\}$ such that $|r_1|\not=|r_2|$. 
For every positive integer $q'$ there exist
 $\hq=\Omega(q')$ sets of positive even integers
$\mc{H}_0,\ldots,\mc{H}_{\hq}$ such that
\begin{renumerate}
 \item $\sigma(\mc{H}_i) = O(\log^3 q')$  for all $i$, 
 \item $\sigma(\mc{H}_i) = \sigma(\mc{H}_j) $  for all $i\not=j$,
 \item $f_{\t,\mc{H}_i}(e_1,e_2,r_1,r_2)\not=f_{\t,\mc{H}_j}(e_1,e_2,r_1,r_2)$ for $i\not=j$,  
 \end{renumerate}
where $f_{\t,\mc{H}}(e_1,e_2,r_1,r_2)$ is from Proposition \ref{lem:otimes_reduction}. 
If additionally $\delta\notin\{-1,1\}$ and $\gamma\not=-\delta$, we have
\begin{renumerate}
 \item[(iv)] $g_{\y,\mc{H}_i}(\gamma,\delta)\not=g_{\y,\mc{H}_j}(\gamma,\delta)$ for $i\not=j$. 
 \item[(v)] $g_{p,\mc{H}_i}(\gamma,\delta)\not=0$,
\end{renumerate}
where $g_{\y,\mc{H}}(e_1,e_2,r_1)$ is from Proposition \ref{prop:S_big_ising}.

The sets $\mc{H}_i$ can be computed in polynomial time in $q'$. 
\end{lem}
\begin{proof}

Let $q=q'\log^3 q'$. First we define sets $\mc{H}_0',\ldots,\mc{H}_{q}'$. We will use these sets to define the desired sets $\mc{H}_0,\ldots,\mc{H}_{\hq}$.

For $i=0,\ldots,q$, let $i[0],\ldots,i[\ell]\in \{0,1\}$ be the binary expansion of $i$
where\footnote{Actually, we will also need  that $\ell$ is larger
than a constant depending on $e_1$, but this is true for large enough values of $q$. } 
$\ell = \lfloor \log q \rfloor$. 
Let $\Delta$ be a  positive even integer to be chosen later. 
Let $\tau\in \{1,2\}$ be such that $|r_\tau| = \max\{|r_1|,|r_2|\}$. 
Then $|r_{\kt}| = \min\{|r_1|,|r_2|\}$. 
%Let $m_0\in \mathbb{N}^+$ be an even number such that for all $h>m_0$, 
%$\frac{1}{2} |r_\tau|^h \leq e_1 |r_1|^h + e_2 |r_2|^h \leq \frac{3}{2} |r_\tau|^h$
%and $e_1 r_1^h + e_2 r_2^h\not=0$. 
Let $m_0$ be an even integer such that $m_0>h_1$ from Proposition \ref{prop:S_big2_ising} and $m_0>h_2$ from Proposition \ref{prop:h_2_ising}.
We choose $\mc{H}_i'$ as follows:
\[
 \mc{H}_i' = \{ m_0+\Delta \lceil \log q \rceil  \cdot (2j + i[j])  : 0\leq j \leq \ell \}\,. 
\]
The sets $\mc{H}_i'$ satisfy:
\begin{enumerate}
\item[a.] they are distinct, 
 \item[b.] they have equal cardinality $\ell$+1,
 \item[c.] they contain only positive even integers between $m_0$\\ and $m_0+\Delta (\log q+1)(2\log q+1)$, and
 \item[d.] for $i,j$ and any $a\in \mc{H}_i'$ and $b\in \mc{H}_j'$, either $a = b$ or \\ $|a-b|\geq \Delta \log q$. 
\end{enumerate}
It is easy to see that $\sigma(\mc{H}_i')=\Omega (\log q)$, $i=0,\ldots,q$. On the other hand,
since all the numbers in each of the $\mc{H}_i'$ are bounded  by $O(\log^2 q)$
and the size of each $\mc{H}_i'$ is $O(\log q)$, we get that $\sigma(\mc{H}_i') = O(\log^3 q)$ for each $i$. 
From this we get that at least $\hq = \Omega\left(\frac{q'\log^3 q' +1}{\log^3 q'}\right) = \Omega(q')$ of the sets $\mc{H}_0',\ldots,\mc{H}_{q}'$
have the same sum value $\sigma(\mc{H}_i')$. 
Let $\{\mc{H}_0,\ldots,\mc{H}_{\hq}\}$ be a subset of $\{\mc{H}_0',\ldots,\mc{H}_{q}'\}$ 
such that all the sets in $\{\mc{H}_0,\ldots,\mc{H}_{\hq}\}$ have the same sum value $\sigma(\mc{H}_i)$. 
We have (i), (ii) and (v) for $\mc{H}_0,\ldots,\mc{H}_{\hq}$.

We now turn to (iii) and (iv). The proofs of (iii) and (iv) are similar but not identical. 

Let $0\leq i\not=j\leq \hq$,
 $\mcHb{i}{j} = \mc{H}_i \setminus \mc{H}_j$
and $\mcHb{j}{i} = \mc{H}_j \setminus \mc{H}_i$. Notice $\mcHb{i}{j}\cap \mcHb{j}{i} = \emptyset$. \\
Let $\sigma=\sigma(\mcHb{i}{j})=\sigma(\mcHb{j}{i})$ and let $d=|\mcHb{i}{j}|=|\mcHb{j}{i}|$. 

\begin{renumerate}
 \item[(iii)] %---------------------------------------------------------------------------------------------------------
We write $f_{\t,\mc{H}_i}$ for short instead of $f_{\t,\mc{H}_i}(e_1,e_2,r_1,r_2)$ in this proof. 
When other parameters are used instead of $e_1,e_2,r_1,r_2$ we write them explicitly. 
Since $f_{\t,\mc{H}_i}=f_{\t,\mcHb{i}{j} } \cdot f_{\t,\mc{H}_{i}\cap \mc{H}_j }$,
$f_{\t,\mc{H}_j}=f_{\t,\mcHb{j}{i} } \cdot f_{\t,\mc{H}_{i}\cap \mc{H}_j }$
and $f_{\t,\mc{H}_{i}\cap \mc{H}_j }\not=0$, it is enough to show that
$f_{\t,\mcHb{i}{j} } - f_{\t,\mcHb{j}{i} }\not=0$.

Since $\sigma(\mcHb{i}{j})=\sigma(\mcHb{j}{i})$ we have 
\[
 f_{\t,\mcHb{i}{j}} = f_{\t,\mcHb{j}{i}}\mbox{ iff }f_{\t,\mcHb{i}{j}}(e_1,e_2,r_2^{-1},r_1^{-1}) = f_{\t,\mcHb{j}{i}}(e_1,e_2,r_2^{-1},r_1^{-1})\,. 
 \]
Hence we can assume from now on that $|r_\tau|>1$ (otherwise we look at $r_1^{-1}$ and $r_2^{-1}$ instead).

For every $\mc{H}$, $f_{\t,\mc{H}} $ can be rewritten as follows:
\[
 f_{\t,\mc{H}} = 
 \prod_{h\in \mc{H}} ( e_\tau r_\tau^h + e_\kt r_{\kt}^h ) = 
 e_\kt^{\ell+1} 
 \sum_{X\subseteq \mc{H}} 
 s_{\mc{H}}(X) \,,
\]
where 
\[
 s_{\mc{H}}(X) = \left(\frac{e_\tau}{e_\kt}\right)^{|X|} r_\tau^{\sigma(X)}
  r_\kt^{\sigma(\mc{H}\setminus X)}\,.
\]
We think of $h\in X$ (respectively $h\in \mc{H}\setminus X)$ as corresponding to $e_\tau r_\tau^h$  (respectively $e_\kt r_\kt^h$).

It suffices to show that
\begin{gather}
 \label{eq:s_sums}
 \sum_{X_1\subseteq \mcHb{i}{j}} s_{\mcHb{i}{j}}(X)
 - \sum_{X_2\subseteq \mcHb{j}{i}} s_{\mcHb{j}{i}}(X) \not=0\,. 
\end{gather}
It holds that
$s_{\mcHb{i}{j}}(\mcHb{i}{j}) = s_{\mcHb{j}{i}}(\mcHb{j}{i}) = 
\left(\frac{e_{\tau}}{e_{3-\tau}}\right)^{\ell+1} r_\tau^{\sigma}$.
Hence, \linebreak 
$s_{\mcHb{i}{j}}(\mcHb{i}{j})$ and  $s_{\mcHb{j}{i}}(\mcHb{j}{i})$ cancel out in Inequality (\ref{eq:s_sums}).
Similarly, \linebreak
$s_{\mcHb{i}{j}}(\emptyset) = s_{\mcHb{j}{i}}(\emptyset) = r_{\kt}^\sigma$ cancel out.
Let $m_1$ be the minimal element in $\mcHb{i}{j} \sqcup \mcHb{j}{i}$. Without loss of generality, assume $m_1\in \mcHb{i}{j}$.
We have
\begin{gather}
 \notag
 s_{\mcHb{i}{j}}(\mcHb{i}{j} \setminus \{m_1\}) = \left(\frac{e_{\tau}}{e_{3-\tau}}\right)^\ell r_\tau^{\sigma - m_1} r_\kt^{m_1} \,.
 %\mbox{ and }
 %\\ 
 %s_{\mcHb{i}{j}}(\{m_1\}) = \left(\frac{e_\tau}{e_{3-\tau}}\right)^1 r_\tau^{ m_1} r_\kt^{\sigma - m_1}\,. \\ 
\end{gather}
$s_{\mcHb{i}{j}}(\mcHb{i}{j} \setminus \{m_1\})$ has the largest exponent of $r_\tau$ out of all the monomials
in both of the sums in Inequality (\ref{eq:s_sums}), and any other exponent of $r_\tau$ is smaller by at least $\Delta \log q$. 
We will show that Inequality (\ref{eq:s_sums}) holds by showing the following:
\begin{gather}
\label{eq:s_sums_2}
|s_{\mcHb{i}{j}}(\mcHb{i}{j} \setminus \{m_1\})|>
\sum_{X\subsetneq \mcHb{i}{j}\setminus \{m_1\}} |s_{\mcHb{i}{j}}(X)|
 + \sum_{X\subsetneq \mcHb{j}{i}} |s_{\mcHb{j}{i}}(X)| \,. 
\end{gather}
Each of the sums in Inequality (\ref{eq:s_sums_2}) has at most $2^{\log q+1}=2q$ monomials corresponding 
to the subsets of $\mcHb{i}{j}$ and $\mcHb{j}{i}$ respectively. The absolute value of each of these monomials
can be bounded from above by $s\cdot |r_\tau|^{\sigma-m_1-\Delta\log q} |r_\kt|^{m_1+\Delta\log q}$,
where $s$ is the maximum of $\left|\frac{e_{\tau}}{e_{3-\tau}}\right|^\ell$ and $1$.
Hence, the right-hand side of Inequality (\ref{eq:s_sums_2}) is at most
\begin{eqnarray} \notag
 & & 4q\cdot s |r_\tau|^{\sigma-m_1-\Delta\log q} |r_\kt|^{m_1+\Delta\log q} = \\ \notag 
 & & 4q s \left(\frac{e_{\tau}}{e_{3-\tau}}\right)^{-\ell}\cdot  \left|\frac{r_\kt}{r_\tau}\right|^{\Delta \log q}
 |s_{\mcHb{i}{j}}(\mcHb{i}{j} - \{m_1\})| \,.
\end{eqnarray}
There exists a number $\Delta'$ which does not depend on $q$ such that\linebreak
$4q s \left(\frac{e_{\tau}}{e_{3-\tau}}\right)^{-\ell} < (\Delta')^{\log q}$
and (iii) follows  by setting $\Delta$ large enough so that
$\Delta' \cdot \left|\frac{r_\kt}{r_\tau}\right|^\Delta <1$.

\begin{samepage}
\
\item[(iv)]%---------------------------------------------------------------------------------------------------------
By Proposition \ref{prop:S_big2_ising}, there exist $u_1,u_2\not=0$ and $w\notin\{-1,0,1\}$ depending on $\gamma,\delta$
for which  it is enough to show that \linebreak
$h_{\y,\mc{H}_i}(u_1,u_2,w) \not= h_{\y,\mc{H}_j}(u_1,u_2,w)$ to get (iv). 
We write $h_{\y,\mc{H}_i}$ for short instead of 
$h_{\y,\mc{H}_i}(u_1,u_2,w)$ in this proof. 
\end{samepage}

Since we have $h_{\y,\mc{H}_i}=h_{\y,\mcHb{i}{j} } \cdot h_{\y,\mc{H}_{i}\cap \mc{H}_j }$,
$h_{\y,\mc{H}_j}=h_{\y,\mcHb{j}{i} } \cdot h_{\y,\mc{H}_{i}\cap \mc{H}_j }$
and \linebreak 
$h_{\y,\mc{H}_{i}\cap \mc{H}_j }\not=0$, it is enough to show that
\[h_{\y,\mcHb{i}{j} } - h_{\y,\mcHb{j}{i} }\not=0\,,\]
i.e.
\[
\displaystyle{\prod_{h\in \mcHb{i}{j}}}
\left(1+
\frac{1}
{u_1 
 +u_2 \cdot w^h} \right) -
\displaystyle{\prod_{h\in \mcHb{j}{i}}}
\left(1+
\frac{1}
{u_1 
 +u_2 \cdot w^h} \right) \not=0\,.
\]
or equivalently, 
\begin{eqnarray}\label{eq:bigProducts_Ising}
\displaystyle{\prod_{h\in \mcHb{i}{j}}}
\left(u_1 
 +u_2 \cdot w^h+1\right)
\prod_{h\in \mcHb{j}{i}}
\left(u_1 
 +u_2 \cdot w^h \right) & - &\\ \notag
\displaystyle{\prod_{h\in \mcHb{j}{i}}}
\left(u_1 
 +u_2 \cdot w^h+1\right)
\prod_{h\in \mcHb{i}{j}}
\left(u_1 
 +u_2 \cdot w^h \right) &\not=& 0
\end{eqnarray}
Consider a product of the form found in Inequality (\ref{eq:bigProducts_Ising}). 
\begin{eqnarray}\notag
\displaystyle{\prod_{h\in \mc{H}_a}}
\left(u_1 
 +u_2 \cdot w^h+1\right)
\prod_{h\in \mc{H}_b}
\left(u_1 
 +u_2 \cdot w^h \right)  & =& \\ \notag
\sum_{X\subseteq \mc{H}_a\cup \mc{H}_b} (u_1+1)^{|\mc{H}_a\setminus X|} u_1^{|\mc{H}_b\setminus X|} w^{\sigma(X)} u_2^{|X|} &&
\end{eqnarray}
Let 
\begin{eqnarray} \notag
p(X)& =&
 \left((u_1+1)^{|\mcHb{i}{j}\setminus X|} u_1^{|\mcHb{j}{i}\setminus X|}  - (u_1+1)^{|\mcHb{j}{i}\setminus X|}  u_1^{|\mcHb{i}{j}\setminus X|}\right)\\ 
\notag & & 
\cdot u_2^{|X|} 
   w^{\sigma(X)} 
\end{eqnarray}
It suffices to show that
\begin{equation}
 \label{eq:sums_p}
 \sum_{X\subseteq \mcHb{i}{j}\cup \mcHb{j}{i}} p(X) \not=0\,.
\end{equation}
 
We have $p(\emptyset)=p(\mcHb{i}{j}\sqcup \mcHb{j}{i})=0$, using that 
$|\mcHb{i}{j}|=|\mcHb{j}{i}|$.
Let $m_1$ be the minimal element in $\mcHb{i}{j} \sqcup \mcHb{j}{i}$. Without loss of generality, assume $m_1\in \mcHb{i}{j}$.
We have
\begin{eqnarray}
 \notag
 \left|p(\mcHb{i}{j}\sqcup \mcHb{j}{i} - \{m_1\})\right| &=& \left|u_2^{2d-1} w^{2\sigma-m_1}\right| \\ \notag
 \left|p(\{m_1\})\right| &=& \left|((u_1+1)u_1)^{d-1} u_2 w^{m_1}\right|
 \end{eqnarray}
The largest exponent of $w$ in Inequality (\ref{eq:sums_p}) is $w^{2\sigma-m_1}$. 
For all other monomials in (\ref{eq:sums_p}), the power of $w$ is smaller by at least $\Delta \log q$.
Similarly,  the smallest exponent of $w$ in Inequality (\ref{eq:sums_p}) is $w^{m_1}$. 
For all other monomials in (\ref{eq:sums_p}), the power of $w$ is larger by at least $\Delta \log q$.

Let $X_0\subseteq \mcHb{i}{j}\sqcup \mcHb{j}{i})$ be maximal with respect to $|p(X_0)|$. 
Since $d\leq \log q+1$, we can choose $\Delta$ large enough so that we have \linebreak
$X_0=\mcHb{i}{j}\sqcup \mcHb{j}{i} - \{m_1\}$ if $|w|>1$ and $X_0=\{m_1\}$ if $|w|<1$.

We have the following:
 \begin{gather}\label{eq:sum_leq_ising}
   \left|\sum_{X\subseteq \mcHb{i}{j}\sqcup \mcHb{j}{i}: X\not=X_0} p(X)\right| < |p(X_0)|
 \end{gather}
implying that Inequality (\ref{eq:sums_p}) holds. 
To see that Inequality (\ref{eq:sum_leq_ising}) holds, note that
\begin{gather} \notag
 \left|\sum_{\substack{X\subseteq \mcHb{i}{j}\sqcup \mcHb{j}{i}:\\ X\not=X_0}} p(X)\right| 
 \leq  2^{2\log q+2} \max_{X\subseteq \mcHb{i}{j}\sqcup \mcHb{j}{i}: X\not=X_0} |p(X)|\\ \notag
\end{gather}
Let 
\[
k(d)= \max\left(\left|(u_1+1)^d\right|,1\right) \cdot \max\left(\left|u_1^d\right|,1\right) \cdot \max\left(\left|u_2^d\right|,1\right)  \,.
\]
Then 
\[
 \left|\sum_{\substack{X\subseteq \mcHb{i}{j}\sqcup \mcHb{j}{i}:\\ X\not=X_0}} p(X)\right| 
\leq
\begin{cases} 
4\cdot 2^{\log q+1} \cdot k(d) \cdot|w|^{2\sigma-m_1-\Delta\log q}, & |w|>1 \\
4\cdot 2^{\log q+1} \cdot k(d) \cdot|w|^{m_1+\Delta \log q}, &  |w|<1 
\end{cases}
\]

So, there exists a constant $c>0$ depending on $u_1,u_2,w$ such that (for large enough values of $q$), 
\begin{gather} \notag
 \left|\sum_{\substack{X\subseteq \mcHb{i}{j}\sqcup \mcHb{j}{i}:\\ X\not=X_0}} p(X)\right| 
 \leq  
\begin{cases} 
c^{\log q} |w|^{2\sigma-m_1-\Delta\log q}, & |w|>1 \\
c^{\log q} |w|^{m_1+\Delta \log q}, &  |w|<1 
\end{cases}
\end{gather}

It remains to choose $\Delta$ large enough so that 
\[
 \begin{cases} 
\frac{c^{\log q}}{u_2^{2d-1}} < |w|^{\Delta \log q}, & |w|>1 \\
 |w|^{\Delta \log q} < \frac{((u_1+1)u_1)^{d-1}}{c^{\log q}}, &  |w|<1 
\end{cases}
\]

%\end{description}
\end{renumerate}

\end{proof}

%&*********************************************************************************************
%***********************************************************************************************
\begin{samepage}
We are now ready to prove Theorems \ref{th:mainC}.
\begin{thm}
Let $(\gamma,\delta)\in\mathbb{Q}^{2}$. If  $\gamma\notin\{-1,0,1\}$ and $\delta\not=0$, then
\begin{renumerate}
\item computing $Z(G;\gamma,\delta)$ is $\spP$-hard, and
\item unless $\spETH$ fails, 
 computing $Z(G;\gamma,\delta)$  requires exponential time in $\frac{n_G}{\log^6 n_G}$.
\end{renumerate}
Otherwise, $Z(G;\gamma,\delta)$ is polynomial-time computable. 
\end{thm}
\end{samepage}

\begin{proof}
 We set $\t=\gamma$ and $\y=\delta$ with  $\gamma\notin\{-1,0,1\}$ and $\delta\not=0$. 
 By abuse of notation we refer to $c_1,c_2,d_1,d_2,\lamb_1,\lamb_2$ from Lemma \ref{lem:b_both} as the values
 they obtain when $\t=\gamma$. 
 Since $\gamma\not=\{-1,0,1\}$, it is easy to verify that the following hold:
 \begin{enumerate}
  %\item[a.] $\gamma^2+\gamma\not=0$, 
  \item[a.] $c_1+d_1,c_2+d_2 \not=0$,
  \item[b.] $\lamb_1,\lamb_2 \not=0$,
  \item[c.] $\lamb_1,\lamb_2 \not=\pm (\gamma^2+\gamma)$, and
  \item[d.] $\lamb_1 \not= \pm \lamb_2$.
 \end{enumerate}

Let $e_{i}=\frac{c_i+d_i}{2\gamma}$
and $r_i = \frac{\lamb_i}{\gamma^2+\gamma}$
for $i=1,2$. 
Let $q'=n_G^2$. 
Let $\mc{H}_0,\ldots,\mc{H}_{\hq}$ be the sets guaranteed in Lemma \ref{lem:dell-lemma} 
with respect to $q',\gamma,\delta,e_1,e_2,r_1,r_2$. 

First we deal with that case that $\gamma\not=-\delta$. We return to $\gamma=-\delta$ later. 

We want to compute the $\hq+1$ values $Z(G\otimes \mc{H}_k;\gamma,1)$.
If $\delta =1$ we simply do it using the oracle to $Z$ at $(\gamma,1)$. If $\delta=-1$ we use Lemma \ref{lem:Zminus1}. 
Otherwise we proceed as follows. 

By Proposition \ref{prop:S_big_ising}, for each $0\leq i,k\leq \hq$,
\begin{gather} \label{eq:big_reduction_ising}
Z(S_{\mc{H}_i}(G\otimes \mc{H}_k);\gamma,\delta) = g_{p,\mc{H}_i}(\gamma,\delta) \cdot Z(G\otimes \mc{H}_k;\gamma,g_{\y,\mc{H}_i}(\gamma,\delta))\,. 
\end{gather}
It is guaranteed in Lemma \ref{lem:dell-lemma} that for $i\not=j$, $g_{\y,\mc{H}_i}(\gamma,\delta)\not=g_{\y,\mc{H}_j}(\gamma,\delta)$.

We want to use Equation (\ref{eq:big_reduction_ising}) to interpolate, for each $0\leq k\leq m_G$,
the univariate polynomials $Z(G\otimes \mc{H}_k;\gamma,\y)$.
We use the fact that the sizes of $G\otimes \mc{H}_k$, and therefore the $\y$-degrees of $Z(G\otimes \mc{H}_k;\gamma,\y)$,
 are at most $O(n_G \log^3 n_G)$, 
Since
$g_{p,\mc{H}_i}(\gamma,\delta)$ is non-zero, we can interpolate in polynomial-time, for each $0\leq k \leq m_G$,
the $m_G+1$ polynomials  
$Z(G\otimes \mc{H}_k;\gamma,\y)$.

So, we computed $Z(G\otimes \mc{H}_k;\gamma,1)$ for $0\leq k\leq \hq$. 
Now we use these values to interpolate $\t$ and get the univariate polynomial $Z(G\otimes \mc{H}_k;\t,1)$. 
By Lemma \ref{lem:otimes_reduction}, 
\[
 Z(G;f_{\t,\mc{H}_k}(e_1,e_2,r_1,r_2),1) = Z(G\otimes \mc{H}_k; \gamma,1) \cdot \left(f_{p,\mc{H}_k}(\gamma)\right)^{-1}\,.
\]
Since $\gamma\notin \{-1,0,1\}$, $f_{p,\mc{H}_k}(\gamma)\not=0$. 
By Lemma \ref{lem:dell-lemma}, $f_{\t,\mc{H}_k}(e_1,e_2,r_1,r_2)$ are distinct and polynomial time computable. 
Hence, the univariate polynomial $Z(G;\t,1)$ can be interpolated.
We get (i) by Proposition \ref{prop:complexity2a}.
Since $Z(-;\gamma,\delta)$ is only queried on graphs $S_{\mc{H}_i}(G\otimes \mc{H}_k)$ of sizes at most $O(n_G \log^6 n_G)$, (ii) holds
by Proposition \ref{prop:complexity2b}. 

Consider the case $\gamma=-\delta$. 
By Proposition \ref{prop:S_big_ising}, for every $G$ we have
\begin{gather}\notag
Z(S_{\{1\}}(G);\gamma,\delta) = (\delta\cdot (1-\delta^2))^{n_G} \cdot Z(G;\gamma,-\delta^2)\,. 
\end{gather}
and the desired hardness results follow by the corresponding for $Z(G;\gamma,-\delta^2)$ (using that $\gamma\not=-(-\delta^2)$, $-\delta^2\not\in\{-1,0,1\}$
and that  $(\delta\cdot (1-\delta^2))^{n_G}$ is non-zero). 

Now we consider the cases where $\gamma \in \{-1,0,1\}$ or $\delta=0$. 
Two cases are easily computed, namely $Z(G;1,\delta)=(1+\delta)^{n_G}$ and $Z(G;\gamma,0)=1$. 

The other two cases follows e.g. from Lemma 6.3 in \cite{ar:GoldbergGroheJerrumThurley09}. 
In that lemma it is shown in particular that partition functions $Z_{A,D}(G)$ with a matrix $A$ of edge-weights
and a diagonal matrix $D$ of vertex weights can be computed in polynomial time if $A$ has rank $1$ or is bipartite with rank $2$. 
For $\gamma=0$ we have 
\begin{eqnarray*}
 A = \left( \begin{array}{cc}
0 & 1 \\
1 & 0 
\end{array} \right)
& &
D = \left( \begin{array}{cc}
\delta & 0 \\
0 & 1 
\end{array} \right)
\end{eqnarray*}
so $A$ is bipartite with rank $2$. 
For $\gamma=-1$ we have 
\begin{eqnarray*}
 A = \left( \begin{array}{cc}
-1 & 1 \\
1 & -1 
\end{array} \right)
& &
D = \left( \begin{array}{cc}
\delta & 0 \\
0 & 1 
\end{array} \right)
\end{eqnarray*}
so $A$ has rank $1$. Note Lemma 6.3 extends to negative values of $\delta$. 
We refer the reader to \cite{ar:GoldbergGroheJerrumThurley09} for details. 
\end{proof}

%-----------------------------------------------------------------------------------------------

\section{Simple Bipartite Planar Graphs}
\label{se:trivariate}

In this section we show that the evaluations of $Z(G;\x,\y,\z)$
are generally $\spP$-hard to compute, even when restricted to simple graphs
which are both bipartite and planar. 
To do so, we use that for $3$-regular graphs, $Z(G;\x,\y,\z)$ is essentially equivalent to $Z(G;\t,\y)$. 
We use a two-dimensional graph transformation 
 $R^{\ell,q}(G)$ which in applied to simple 
 $3$-regular bipartite planar graphs 
 and emits simple bipartite planar graphs in order to interpolate $Z(G;\t,\y)$. 
 
 \subsection{Definitions}
 The following is a variation of {\em $k$-thickening} for simple graphs:
\begin{defn}
[$k$-Simple Thickening]
Given $\ell\in\mathbb{N}^{+}$ and a graph
$H$, we define a graph $STh^{\ell}(H)$
%$\widehat{H}^{\ell}$ 
as follows. 
For every edge $e=(u,w) $
in $E(H)$, we add $4\ell$ new vertices $v_{e,1},\ldots,v_{e,4\ell}$
to $H$. For each $v_{e,i}$, we add two new edges
$(u,v_{e,i})$ and $(w,v_{e,i})$. 
Finally, we remove the edge $e$ from the graph. 
Let $N_{\ell}(e)^{+}$ denote the subgraph of $STh^{\ell}(H)$ induced by 
the set of vertices $\left\{ v_{e,1},\ldots,v_{e,4\ell},u,v\right\} $.
\end{defn}

The graph transformation used in the hardness proof is the following:
\begin{defn}
[$R^{\ell,q}(G)$]
Let $G$ be a graph. For each $w\in V(G)$,
let $G_{w}^{q}=(V_{w}^{q},E_{w}^{q})$ be a new copy of the star with $2q$ leaves.
Denote by $c_{w}$ the center of the star $G_{w}^{q}$. 
%The vertices
%of all the graphs $G$ and $G_{w}$ for $w\in V(G)$ are disjoint.
Let $R^{\ell,q}(G)=(V_{R}^{\ell,q},E_{R}^{\ell,q})$ be the graph
obtained from the disjoint union of 
$STh^{\ell}(G)$ and $STh^{\ell}(G_w^q)$
for all $w\in V(G)$ 
by identifying $w$ and $c_{w}$ for all $w\in V(G)$.
\end{defn}

\begin{remarks} ~
\begin{renumerate}
 \item
 The construction of $R^{\ell,q}(G)$ can also be described as follows.
 Given $G$, we attach $2q$ new vertices to each vertex $v$ of $V(G)$ to obtain
 a new simple graph $G'$. Then, $R^{\ell,q}(G)=STh^{\ell}(G')$. 
 \item
 For every simple planar graph $G$ and $\ell,q\in \mathbb{N}^{+}$, 
 $R^{\ell,q}(G)$ is a simple bipartite planar bipartite graph 
 with $n_R$ vertices and $m_R$ edges, where 
 $n_{R}=n_{G}(1+2q(1+4\ell))+4\ell m_{G}$ and 
 and $m_{R}=8\ell m_{G}+16\ell qn_{G}$.
 \end{renumerate}
\end{remarks}
%-------------
Figure \ref{fig:Figure_R} shows the graph $R^{1,2}(P_2)$. 
\begin{figure}
\caption{\label{fig:Figure_R} The construction of the graph $R^{\ell,q}(P_2)$ for $\ell=1$ and $q=2$, where $P_2$ is the path with two vertices and one edge.  }
\begin{center}
\includegraphics[scale=0.5]{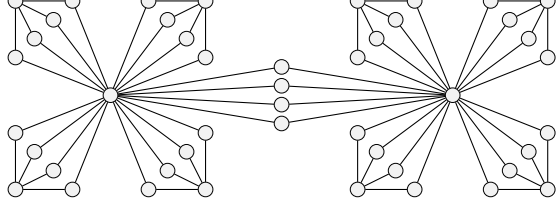}
\end{center}
\end{figure}

%--------------

In the following it is convenient to consider a multivariate version of $Z(G;\x,\y,\z)$ denoted $Z(G;\bar{\x},\bar{\y}.\bar{\z})$. 
This approach was introduced for the Tutte polynomial by A. Sokal, see \cite{ar:Sokal05}. 
$Z(G;\bar{\x},\bar{\y}.\bar{\z})$ has indeterminates which 
correspond to every $v\in V(G)$ and every $e\in E(G)$. 
\begin{defn}
 Let $\bar{\x}=\left(\x_e : e\in E(G)\right)$, $\bar{\y}=\left(\y_u : u\in V(G)\right)$ 
 and \linebreak
$\bar{\z}=\left(\z_e : e\in E(G)\right)$
 be tuples
 of distinct indeterminates. Let 
 \[
  Z(G;\bar{\x},\bar{\y},\bar{\z}) = 
  \sum_{S\subseteq V(G)} \left(\prod_{e\in E_G(S)} \x_e\right) \left( \prod_{u\in S} \y_u\right)
  \left( \prod_{e\in E_G(\bar{S})} \z_e\right)  \,. 
 \]
\end{defn}
We may write $\x_{w,v}$ and $\z_{w,v}$ instead of $\x_e$ and $\z_e$ for an edge $e=(w,v)$. 
Clearly, by setting $\x_e=\x$ and $\z_e=\z$ for every $e\in E(G)$, and $\y_u=\y$ for every $u\in V(G)$
we get $Z(G;\bar{\x},\bar{\y},\bar{\z}) = Z(G;\x,\y,\z)$. 

We furthermore define a variation of $Z(G;\bar{\x},\bar{\y},\bar{\z})$ obtained
by restricting the range of the summation variable as follows:
\begin{defn}
Given a graph $H$ and $B,C\subseteq V(H)$ with $B$ and $C$ disjoint,
let
\begin{eqnarray}
\label{eq:ZprimeBC}
& & Z(H,B,C;\bar{\x},\bar{\y},\bar{\z})=
\\ \notag & & \sum_{A\,:\, B\subseteq A\subseteq V(H)\,,\, A\cap C=\emptyset}
\left(\prod_{e\in E_G(A)} \x_e\right) \left( \prod_{u\in A\setminus B} \y_u\right)
  \left( \prod_{e\in E_G(\bar{A})} \z_e\right) 
\end{eqnarray}
where the summation is over all $A\subseteq V(H)$, such that $A$
contains $B$ and is disjoint from $C$. 
\end{defn}
We have $Z(H,\emptyset,\emptyset;\bar{\x},\bar{\y},\bar{\z})=Z(H;\bar{\x},\bar{\y},\bar{\z})$. 

%Note that in Appendix \AppendixCW\ we use a different multivariate version of $Z(G;\x,\y,\z)$. 

%------------------

\subsection{Lemmas, statement of Theorem \ref{th:mainB} and its proof}
For every edge $e\in E(G)$ between $u$ and $v$, 
let
\[
\omega_{1}(e,S)=Z(N_{\ell}(e)^{+},S\cap\{u,v\},\{u,v\}\setminus S;\bar{\x},\bar{\y},\bar{\z})\,,
\]
and for every vertex $w\in V$,  let
\[
\omega_{2}(w,S)=Z(STh^{\ell}(G_w^q),S\cap\{w\},\{w\}\setminus S;\bar{\x},\bar{\y},\bar{\z})\,.
\]
 
 Let 
 $$
 \omega_{1}(S)=\prod_{e\in E(G)} \omega_{1}(e,S) \mbox{\ \ \  and \ \ \ }
 \omega_{2}(S)=\prod_{w\in V(G)}\omega_{2}(w,S)\,.
 $$
 Let $\omega_{i,\mathrm{\mathrm{triv}}}(S)$ for $i=1,2$ be the polynomials in $\x$, $\y$ and $\z$ obtained from 
 $\omega_i(S)$ by setting $\x_e=\x$ and $\z_e=\z$ for every $e\in E^{\ell,q}$ and
 $\y_v=\y$ for every $v\in V^{\ell,q}_R$. 

\begin{lem}
 \label{lem:SumSThk}
 \[
Z(R^{\ell,q}(G);\x.\y.\z)=
\sum_{S\subseteq V(G)}
\omega_{1,\mathrm{triv}}(S) \cdot \omega_{2,\mathrm{\mathrm{triv}}}(S) \cdot \y^{|S|}\,.
\]
\end{lem}
\begin{proof}
Each edge of $R^{\ell,q}(G)$ is either
contained in some $N_{\ell}(e)^{+}$ for $e\in E(G)$, or in some
$STh^{\ell}(G_w^q)$ for $w\in V(G)$. Hence, 
by the definitions of 
$Z(R^{\ell,q}(G);\bar{\x},\bar{\y},\bar{\z})$, $\omega_1(S)$ and $\omega_2(S)$, 
\[
%\label{eq:SumSThk}
\notag
Z(R^{\ell,q}(G);\bar{\x},\bar{\y},\bar{\z})=
\sum_{S\subseteq V(G)}
\omega_1(S) \cdot \omega_2(S) \cdot \prod_{w\in S} \y_w
%y^{|S|}\left(\prod_{e\in E(G)}\omega_{1}(e,S)\prod_{w\in V(G)}\omega_{2}(v,S)\right)
\]
holds and the lemma follows.
 \end{proof}

\begin{lem}
\label{lem:omega1}
Let $e=(u,w)$ be an edge of $G$. Then
\[
  \omega_{1,\mathrm{triv}}(e,S) = 
 \begin{cases}
 (\y+\z^2)^{4\ell} & |\{u,v\}\cap S| = 0 \\
  (\x\y+\z)^{4\ell} & |\{u,v\}\cap S| = 1 \\
  (\y\x^2+1)^{4\ell} & |\{u,v\}\cap S| = 2 
 \end{cases}
\]
\end{lem}
\begin{proof}
 The value of $\omega_{1}(e,S)$ depends only on
whether $u,w\in S$. Consider $A\subseteq V(N_{\ell}(e)^{+})$ which
satisfies the summation conditions in Equation (\ref{eq:ZprimeBC})
for $Z(N_{\ell}(e)^{+},S\cap \{u,w\},\{u,w\}\setminus S;\x,\y,\z)$. 
\begin{renumerate}
\item If $w\in S$ and $u\notin S$:
Exactly one edge $e'$ incident to $v_{e,i}$ crosses the cut $[A,\bar{A}]_{N_{\ell}(e)^{+}}$.
The other edge $e''$ incident to $v_{e,i}$ belongs to $E(A)$ or $E(\bar{A})$, depending on
whether $v_{e,i}\in A$. 
We get:
\begin{gather}
\notag
\label{eq:with_edge_removed}
\omega_{1}(e,S)=\prod_{i=1}^{4\ell} 
(\x_{v_{e,i},w} \y_{v_{e,i}} + \z_{v_{e,i},u})\,. 
\end{gather}
\item If $w\notin S$ and $u\in S$: this case is similar to the previous case, and we get
\begin{gather}
\notag
\omega_{1}(e,S)=\prod_{i=1}^{4\ell} 
(\x_{v_{e,i},u} \y_{v_{e,i}} + \z_{v_{e,i},w})\,. 
\end{gather}
\item If $w,u\in S$: For each $v_{e,i}$, either $v_{e,i}\in A$, in which
case both edges $(v_{e,i},w)$ and $(v_{e,i},u)$ are in $E(A)$, or
$v_{e,i}\notin S$, and both edges  $(v_{e,i},w)$ and $(v_{e,i},u)$ cross the
cut. We get:
\[\omega_{1}(e,S)=\prod_{i=1}^{4\ell} (\y_{v_{e,i}} \x_{v_{e,i},u} \x_{v_{e,i},w} + 1 )\,.\]
\item If $w,u\notin S$: For each $v_{e,i}$, either $v_{e,i}\in S$ and
then both edges incident to $v_{e,i}$ cross the cut, or $v_{e,i}\notin S$
and none of the two edges cross the cut. We get:
\[\omega_{1}(e,S)=\prod_{i=1}^{4\ell} (\y_{v_{e,i}} + \z_{v_{e,i},w} \z_{v_{e,i},u})\,.\] 
\end{renumerate}
The lemma follows by setting $\x_e=\x$ and $\z_e=\z$ 
for every edge $e$ and $\y_u=\y$ for every vertex $u$. 
\end{proof}

\begin{lem}
\label{lem:omega2}
Let
\begin{eqnarray*}
g_{\ell,q}(\x,\y,\z)&=&
\y\cdot\left( \y \x^2+1 \right)^{4\ell} + 
\left(\y \x+\z\right)^{4\ell}\\
h_{\ell,q}(\x,\y,\z)&=&
\left(\y+\z^2\right)^{4\ell} +
\y\cdot\left(\y\x+\z\right)^{4\ell} 
\,.
\end{eqnarray*}
Let $w$ be a vertex of $G$. Then
\begin{eqnarray*}
\omega_{2,\mathrm{triv}}(w,S) & = &
 \begin{cases}
  \left(g_{\ell,q}(\x,\y,\z)\right)^{2q} & w\in S \\
  \left(h_{\ell,q}(\x,\y,\z)\right)^{2q} & w\notin S
 \end{cases}
\end{eqnarray*}

\end{lem}
\begin{proof}
Consider $A$
which satisfies the summation
conditions in Equation (\ref{eq:ZprimeBC}) for 
$Z(STh^{\ell}(G_w^q),S\cap\{w\},\{w\}\setminus S;\bar{\x},\bar{\y},\bar{\z})$. 
\begin{renumerate}
\item If $w\in S$ (or, equivalently, $c_w\in A$): Let $u\in V_{w}^{q}\setminus\left\{ c_{w}\right\} $
and $e=\{u,c_{w}\}$. If $u\in A$, then the vertices $u$ and $v_{e,1},\ldots,v_{e,4\ell}$
contribute 
\[
 \y_u\prod_{i=1}^{4\ell} (\y_{v_{e,i}} \x_{v_{e,i},w} \x_{v_{e,i},u} + 1)\,.
\]
Otherwise, if $u\notin A$, then the vertices
$u$ and $v_{e,1},\ldots,v_{e,4\ell}$ contribute 
\[
\prod_{i=1}^{4\ell} (\y_{v_{e,i}} \x_{v_{e,i},w} + \z_{v_{e,i},u})\,.
\]
Hence, $\omega_{2}(w,S)$ equals in this case
\[
\prod_{u\in V^{q}_w} 
\left( 
\y_u\prod_{i=1}^{4\ell} (\y_{v_{e,i}} \x_{v_{e,i},w} \x_{v_{e,i},u} + 1)
+
\prod_{i=1}^{4\ell} (\y_{v_{e,i}} \x_{v_{e,i},w} + \z_{v_{e,i},u})
\right)\,.
\]
\item If $w\not\in S$ (or, equivalently, $c_w\notin A$): Let $u\in V_{w}^{q}\setminus\left\{ c_{w}\right\} $
and $e=\{u,c_{w}\}$. If $u\in A$, then the vertices $u$ and $v_{e,1},\ldots,v_{e,4\ell}$
contribute 
\[
\y_u\prod_{i=1}^{4\ell} (\y_{v_{e,i}} \x_{v_{e,i},u} + \z_{v_{e,i},w})\,.
\]
Otherwise, if $u\notin A$, then the vertices
$u$ and $v_{e,1},\ldots,v_{e,4\ell}$ contribute 
\[
\prod_{i=1}^{4\ell} (\y_{v_{e,i}} + \z_{v_{e,i},w} \z_{v_{e,i},u})\,.
\]
Hence, $\omega_{2}(w,S)$ equals in this case
\[
\prod_{u\in V^{q}_w} 
\left( 
\prod_{i=1}^{4\ell} (\y_{v_{e,i}} \x_{v_{e,i},u} + \z_{v_{e,i},w})
+
\prod_{i=1}^{4\ell} (\y_{v_{e,i}} + \z_{v_{e,i},w} \z_{v_{e,i},u})
\right)\,.
\]
\end{renumerate}
The lemma follows by setting $\x_e=\x$ and $\z_e=\z$
for every edge $e$ and $\y_u=\y$ for every vertex $u$. 
\end{proof}

\begin{lem}
\label{le:R} If $G$ is $d-$regular, then
\begin{eqnarray} \notag
& &
f_{p,R}(\x,\y,\z,\ell,q)\cdot Z(G;f_{\t,R}(\x,\y,\z,\ell),f_{\y,R}(\x,\y,\z,\ell,q))= 
 Z(R^{\ell,q}(G);\x,\y,\z)
\end{eqnarray}
where
\begin{eqnarray*}
 f_{p,R}(\x,\y,\z,\ell,q) &=& \left(h_{\ell,q}(\x,\y,\z)\right)^{2qn_{G}}(\y+\z^2)^{2\ell dn_{G}}\,,\\
f_{\t,R}(\x,\y,\z,\ell)&=&
\left(\frac{(\y\x+\z)^2}{(\y\x^2+1)(\y+\z^{2})}\right)^{2\ell}\,,\\
f_{\y,R}(\x,\y,\z,\ell,q)&=&
\y\cdot\left(\frac{\y\x^2+1}{\y+\z^2}\right)^{2\ell d}
\left(\frac{g_{\ell,q}(\x,\y,\z)}{h_{\ell,q}(\x,\y,\z)}\right)^{2q}\,.
\end{eqnarray*}
\end{lem}

\begin{proof}
We want to rewrite $Z(R^{\ell,q}(G);\bar{\x},\bar{\y},\bar{\z})$ as a sum over subsets
$S$ of vertices of $G$. Using Lemma \ref{lem:SumSThk},  in order to 
compute 
$Z(R^{\ell,q}(G);\x,\y,\z)$ we first need to
find $\omega_{1,\mathrm{triv}}(S)$ and $\omega_{2,\mathrm{triv}}(S)$.
Using Lemma \ref{lem:omega2}, 
$\omega_{2,\mathrm{triv}}(S)$ is given by
\[
\omega_{2,\mathrm{triv}}(S)=
\left(g_{\ell,q}(\x,\y,\z)\right)^{2q|S|}\cdot\left(h_{\ell,q}(\x,\y,\z)\right)^{2qn_{G}-2q|S|}\,.
\]
In order to compute $\omega_{1,\mathrm{triv}}(S)$, consider $S\subseteq V(G)$. 
Since $G$ is $d$-regular,  
the number of
edges contained in $S$ is $\frac{1}{2}(d\cdot|S|-|[S,\bar{S}]_{G}|)$,
and the number of edges contained in $\bar{S}$ is 
$\frac{1}{2}(dn_{G}-d\cdot|S|-|[S,\bar{S}]_{G}|)$.
Hence, by Lemma \ref{lem:omega1}, $\omega_{1,\mathrm{triv}}(S)$ is given by
\[
\omega_{1,\mathrm{triv}}(S)=(\x\y+\z)^{4\ell|[S,\bar{S}]_{G}|}
(\y\x^2+1)^{4\ell\cdot\frac{d\cdot|S|-|[S,\bar{S}]_{G}|}{2}}
(\y+\z^2)^{4\ell\cdot\frac{dn_{G}-d\cdot|S|-|[S,\bar{S}]_{G}|}{2}}
\,.
\]
Using Lemma \ref{lem:SumSThk},
\begin{equation}
\notag
Z(R^{\ell,q}(G);\x,\y,\z)=
\sum_{S\subseteq V(G)}
\omega_{1,\mathrm{triv}}(S) \cdot \omega_{2,\mathrm{triv}}(S) \cdot \y^{|S|} 
%y^{|S|}\left(\prod_{e\in E(G)}\omega_{1}(e,S)\prod_{w\in V(G)}\omega_{2}(v,S)\right)
\end{equation}
which is equal to $(\y+\z^{2})^{4\ell\cdot\frac{dn_{G}}{2}}$ times
\begin{equation}
\label{eq:ZprimeNoPref}
\sum_{S\subseteq V(G)}
\left(\frac{(\y\x+\z)^2}{(\y\x^2+1)(\y+\z^{2})}\right)^{2\ell|[S,\bar{S}]_G|}
\left(\y\cdot\left(\frac{\y\x^2+1}{\y+\z^2}\right)^{2\ell d}\right)^{|S|}
\cdot \omega_{2,\mathrm{triv}}(S)\,.
\end{equation}

Plugging the expression for $\omega_{2,\mathrm{triv}}(S)$ in Equation (\ref{eq:ZprimeNoPref}), we get that
$Z(R^{\ell,q}(G);\x,\y,\z)$ equals 
 $f_{p,R}(\x,\y,\z,\ell,q)$ times 
\begin{gather}\notag
\sum_{S\subseteq V(G)}
\left(\frac{(\y\x+\z)^2}{(\y\x^2+1)(\y+\z^{2})}\right)^{2\ell|[S,\bar{S}]_G|}
%\left((\y\x+\z)\left(\frac{\y+\z^{2}}{\y\x^2+1}\right)^{\frac{1}{2}}\right)^{4\ell|[S,\bar{S}]_{G}|}
\left(\y\cdot\left(\frac{\y\x^2+1}{\y+\z^2}\right)^{2\ell d}
\left(\frac{g_{\ell,q}(\x,\y,\z)}{h_{\ell,q}(\x,\y,\z)}\right)^{2q}
\right)^{|S|}
\end{gather}
and the lemma follows. 
\end{proof}

\begin{lem}\label{lem:growth}
Let $e\in\mathbb{Q}\backslash\left\{ -1,0,1\right\} $ and let $a,b,c>0$
and $b\not=c$. Then there is $c_{1}\in\mathbb{N}$ for which the sequence
\begin{gather}\label{eq:lem:growth}
h(\ell) =\frac{e\cdot b^{\ell}+a^{\ell}}{c^{\ell}+e\cdot a^{\ell}}
\end{gather}
is strictly monotone increasing or decreasing for $\ell\geq c_1$. 
\end{lem}
\begin{proof}
 $h(\ell)$ can be rewritten as 
\begin{eqnarray*}
h(\ell) &=&\frac{e\cdot \tb^{\ell}+1}{\tc^{\ell}+e}
\end{eqnarray*}
by dividing both the numerator and the denominator of the right-hand side of Equation (\ref{eq:lem:growth}) 
by $a^\ell$ and setting $\tb=\frac{b}{a}$ and $\tc=\frac{c}{a}$.
We have $\tb\not=\tc$ and $\tb,\tc>0$. 

Let $h(x)=\frac{e\cdot \tb^{x}+1}{\tc^{x}+e}$. 
The derivative of $h(x)$ is given by
\begin{eqnarray}
h'(x)&=&\frac{e\ln \tb\cdot \tb^{x}\left(\tc^{x}+e\right)-\ln \tc\cdot \tc^{x}\left(e\cdot \tb^{x}+1\right)}{(\tc^x +e)^2} \notag \\ \label{eq:gro_der}
&=&\frac{e^2\ln \tb\cdot \tb^{x}-\ln \tc\cdot \tc^{x}+e (\ln \tb-\ln \tc)\tb^{x}\tc^{x}}{(\tc^x +e)^2 }
\end{eqnarray}
The denominator of $h'(x)$ is non-zero for large enough $x$. Therefore, there exists $x_0$ such that 
$h'(x)$ is continuous on $[x_0,\infty)$, so it is enough to show that $h'(x)\not=0$ for all large enough $x$ to get the desired result.

If $\tb=1$ then $ (\tc^x +e)^2 h'(x)=-(1+e)\ln \tc \cdot \tc^x$, and if $\tc=1$ then $(\tc^x +e)^2 h'(x)=(e^2+e)\ln \tb \cdot \tb^x$.
In both cases $h'(x)$ is non-zero, using that $\tb\not=\tc$ and $\tb,\tc>0$. 

Otherwise, $\tb$, $\tc$ and $\tb\tc$ are distinct. 
Let $A_1=\{\tb^x,\tc^x,\tb^x\tc^x\}$.
Let $A_2$ be the subset of $A_1$ which contains the functions of $A_1$
which have non-zero coefficients in Equation (\ref{eq:gro_der}).
Note $\tb^x\tc^x$ belongs of $A_2$. 
There is a function in $A_2$ which dominates the other functions of $A_2$.
This implies that $h'(x)$ is non-zero for large enough values of $x$. 
\end{proof}

Theorem \ref{th:mainB} is now given precisely and proved: 
\begin{thm}
\label{th:IsingSimple} 
For all $(\gamma,\delta,\epsilon)\in\mathbb{Q}^{3}$
such that 
\begin{renumerate}
 \item \label{newreq3} $\delta\not=\{-1,0,1\}$,
 \item \label{newreq1} $\delta+\epsilon^2\notin\{-1,0,1\}$,
 \item \label{newreq4} $\delta+\epsilon^2\not=\pm (\delta\gamma^2+1)$,
 \item \label{newreq2} $\delta\gamma^{2}+1\not=0$,
 \item \label{newreq5} $\gamma\delta+\epsilon\not=0$, and
 \item \label{newreq6} $(\gamma \delta+\epsilon)^4 \not=\left(\delta\gamma^2+1\right)^2 \left( \delta+\epsilon^2\right)^{2}$.
\end{renumerate}
$Z(-;\gamma,\delta,\epsilon)$ is $\spP$-hard on simple bipartite planar graphs. 
\end{thm}

\begin{proof}
We will show that, on $3-$regular bipartite planar graphs $G$,  the polynomial
$Z(G;\t,\y)$ is polynomial-time computable using oracle calls to $Z(-;\gamma,\delta,\epsilon)$.
The oracle is only queried with input of simple bipartite planar graphs. 
Using Proposition \ref{prop:complexity2a}, computing $Z(G;\t,\y)$ is
$\spP$-hard on $3$-regular bipartite planar graphs.

Using (\ref{newreq3}) and (\ref{newreq1})  it can be verified that
there exists $c_0 \in \mathbb{N}^{+}$ such that 
for all $\ell\geq c_0$ and $q\in\mathbb{N}^{+}$, $f_{p,R}(\gamma,\delta,\epsilon,2\ell,q)\not=0$. 
We can use Lemma \ref{le:R} to manufacture, in polynomial-time, evaluations of $Z(G;\t,\y)$
that will be used to interpolate  $Z(G;\t,\y)$.

Let $\ell \geq c_0$ and let 
\[
E_{\y,1}=\frac{\delta \gamma^{2}+1}{\delta+\epsilon^{2}}
\mbox{\  \ \ and \ \ \ }
E_{\y,2,\ell}=
\frac{\delta (\delta\gamma^2+1)^{4\ell} + (\gamma\delta+\epsilon)^{4\ell} }
{ (\delta+\epsilon^2)^{4\ell} +  \delta(\gamma\delta+\epsilon)^{4\ell}} \,.
\]
We have that
$f_{\y,R}(\gamma,\delta,\epsilon,\ell,q)=\delta\left(E_{\y,1}\right)^{2d\ell}\left(E_{\y,2,\ell}\right)^{2q}$.
Using (\ref{newreq2}) we have $E_{\y,1}\not=0$. 

Look at $E_{\y,2,\ell}$ as a function of $\ell$.
Using (\ref{newreq3}), (\ref{newreq1}), (\ref{newreq4}) and (\ref{newreq2}) 
and Lemma \ref{lem:growth}
with $a=(\gamma\delta+\epsilon)^4$, $b=(\delta \gamma^2+1)^4$, $c=(\delta +\epsilon^2)^4$ and $e=\delta$,
there exists $c_1$ such that 
$E_{\y,2,\ell}$ is strictly monotone increasing or decreasing. 
Hence,
there exists $c_2\geq c_1$ such that, for every $\ell\geq c_2$, $E_{\y,2,\ell}\notin\{-1,0,1\}$.
Moreover, $c_2=c_2(\gamma,\delta,\epsilon)$ is a function of $\gamma$, $\delta$ and $\epsilon$.

We get that for $q_{1}\not=q_{2}\in [n_{G}+1]$ and $\ell>c_2$,
$\left(E_{\y,2,\ell}\right)^{2q_{1}}\not=\left(E_{\y,2,\ell}\right)^{2q_{2}}$.
Since $\delta\left(E_{y,1}\right)^{2d\ell}$ is not equal to $0$ and
does not depend on $q$, we get that for $q_{1}\not=q_{2}\in [n_{G}+1]$,
$f_{\y,R}(\gamma,\delta,\epsilon,\ell,q_{1})\not=f_{\y,R}(\gamma,\delta,\epsilon,\ell,q_{2})$. 

For every $\ell\in [m_{G}+c_2+1]\setminus [c_2]$, we can interpolate in polynomial-time 
the univariate polynomial 
$Z(G;f_{\t,R}(\gamma,\delta,\epsilon,\ell),\y)$. 
Then, we can use the polynomial 
 $Z(G;f_{\t,R}(\gamma,\delta,\epsilon,\ell),\y)$
to compute 
$Z(G;f_{\t,R}(\gamma,\delta,\epsilon,\ell),j)$ for every $\ell\in [m_{G}+c_2+1]\setminus [c_2]$
and every $j\in [n_{G}+1]$. 
Let 
\[
E_{\t}=
 \left( \frac{(\gamma \delta+\epsilon)^2}{(\delta\gamma^2+1)(\delta+\epsilon^2)}\right)^{2}
\]
and it holds that $f_{\t,R}(\gamma,\delta,\epsilon,\ell,q)=\left(E_{\t}\right)^{\ell}$.
Clearly, $E_{\t}\not=-1$ and, by 
(\ref{newreq5}) and (\ref{newreq6}), $E_{\t}\notin\{0,1\}$. 
Hence, for every $\ell_{1}\not=\ell_{2}\in\mathbb{N}^{+}$ we have 
$f_{\t,R}(\gamma,\delta,\epsilon,\ell_{1})\not=f_{\t,R}(\gamma,\delta,\epsilon,\ell_{2})$.
Therefore, we can compute the value of the bivariate polynomial $Z(G;\t,\y)$
on a grid of points of size 
$(m_{G}+1)\times (n_{G}+1)$ in polynomial-time
using the oracle, and use them to interpolate $Z(G;\t,\y)$. \\ 
~
\end{proof}

%---------------------------------------------------------------------------------------------------------------

\section{Computation on Graphs of Bounded Clique-width}
\label{se:FPPT}
In this section we prove Theorem \ref{th:mainD}. 
Let $G$ be a graph and let $cw(G)$ be its clique-width.  
As discussed in Section \ref{se:cliquewidth},
a $k$-expression $t(G)$ for $G$ with
$k\leq 2^{3cw(G)+2}-1$
can be computed in $\FPT$-time.
Let $\bar{c}=(c_v: v\in V(G))$ be the labels from $[k]$ associated with the vertices
of $G$ by $t(G)$. 
We will show how to compute a multivariate polynomial $\Zc(G,\bar{c};\bar{\x},\bar{\y},\bar{\z})$ with indeterminate
set 
\[
\left\{ \x_{\{i,j\}},\y_{i},\z_{\{i,j\}}\mid i,j\in[k]\right\} 
\]
to be defined below. Note it is not the same multivariate polynomial as in Section \ref{se:trivariate}. 
For simplicity of notation we write e.g. $\x_{i,j}$ or $\x_{j,i}$
for $\x_{\{i,j\}}$.
The multivariate polynomial $\Zc(G,\bar{c};\bar{\x},\bar{\y},\bar{\z})$
is defined as
\begin{equation}
\label{eq:Ztag}
\sum_{S\subseteq V(G)}\,
\left(\prod_{v\in S}\y_{c_{v}}\right)
\left(\prod_{(u,v)\in  E_G(S)}\x_{c_{u},c_{v}}\right)
\left(\prod_{(u,v)\in  E_G(\bar{S})}\z_{c_{u},c_{v}}\right)\,.
\end{equation}
The left-most product in Equation (\ref{eq:Ztag}) is over all vertices $v$ in $S$.  
The two other products are over all edges in $E_G(S)$ and $E_G(\bar{S})$ respectively.
It is not hard to see that
$Z(G;\x,\y,\z)$ is obtained from  $\Zc(G,\bar{c};\bar{\x},\bar{\y},\bar{\z})$ by
substituting all the indeterminates $\x_{i,j}$, $\y_i$ and $\z_{i,j}$ by three indeterminates,
$\x$, $\y$ and $\z$, respectively. 

Given tuples of natural numbers $\bar{\za}=(\za_{i}\,:\, i\in[k])$, 
$\bar{\zb}=(\zb_{i,j}\,:\, i,j\in[k])$ and 
$\bar{\zc}=(\zc_{i,j}\,:\, i,j\in[k])$,
we denote by $t_{\bar{\za},\bar{\zb},\bar{\zc}}(G)$ the coefficient of the monomial
\[
\prod_{i\in[k]}\y_{i}^{\za_{i}}\prod_{i,j\in k}\x_{i,j}^{\zb_{i,j}} \z_{i,j}^{\zc_{i,j}} 
\]
in $\Zc(G;\bar{\x},\bar{\y},\bar{\z})$. We call a triple $(\bar{\za},\bar{\zb},\bar{\zc})$
{\em valid} if $\za_{1}+\ldots+\za_{k}\leq n_G$ and, for all $i,j\in[k]$,
%\[
$\zb_{i,j},\zc_{i,j} \leq m_G$. 
%\begin{cases}
%\max\{m_{G},\za_{i}\cdot \za_{j}\}                    & \mbox{ if } i\not=j \\
%\max\left\{ m_{G},\binom{\za_{i}}{2}\right\}  & \mbox{ if } i=j
%\end{cases}
%\mbox{ and }
%\zc_{i,j} \leq
%\begin{cases}
%\max\{m_{G},\za_{i}\cdot \za_{j}\}                    & \mbox{ if } i\not=j \\
%\max\left\{ m_{G},\binom{\za_{i}}{2}\right\}  & \mbox{ if } i=j
%\end{cases}
%\]
If $(\bar{\za},\bar{\zb},\bar{\zc})$ is not valid, then
$t_{\bar{\za},\bar{\zb},\bar{\zc}}(G)=0$. Therefore, to determine the polynomial $\Zc(G;\bar{\x},\bar{\y},\bar{\z})$
we need only to find $t_{\bar{\za},\bar{\zb},\bar{\zc}}(G)$ for all valid triples
$(\bar{\za},\bar{\zb},\bar{\zc})$.

The $t_{\bar{\za},\bar{\zb},\bar{\zc}}(G)$ form an $\left(k+2k^2\right)$-dimensional
array with $\left(\max\left\{ n_{G},m_{G}\right\}\right)^{k+2k^2}$ integer entries.
Each entry in this table can be bounded from above by $2^{n_{G}}$
and thus can be written in polynomial space, so the size of the table
is of the form $n_G{}^{p_1(cw(G))}$, where $p_1$ is a function of $cw(G)$ 
which does not depend on $n_G$.

We compute $\Zc(G,\bar{c};\bar{\x},\bar{\y},\bar{\z})$ of $G$ by dynamic programming
on the structure of the $k$-expression of $G$. 
\begin{samepage}
\begin{algorithm}
\label{alg:cliquewidth}~
\begin{enumerate}
\item If $(G,i)$ is a singleton of any color $i$, $\Zc(G,\bar{c};\bar{\x},\bar{\y},\bar{\z})=1+\y_{i}$. 
\item If $(G,\bar{c})$ is the disjoint union of $(H,\bar{c}_{H_1})$ and $(H_2,\bar{c}_{H_2})$, then 
\[
 \Zc(G,\bar{c};\bar{\x},\bar{\y},\bar{\z})=\Zc(H_{1},\bar{c}_{H_1};\bar{\x},\bar{\y},\bar{z})\cdot \Zc(H_{2},\bar{c}_{H_2};\bar{\x},\bar{\y},\bar{z})\,.
 \]
\item If $(G,\bar{c})=\eta_{p,r}(H,\bar{c}_H)$: let $\zd_{r}$ and $\zd_{p}$ be the number of
vertices of colors $r$ and $p$ in $H$, respectively. 
\begin{enumerate}
\item  For every valid $(\bar{\za},\bar{\zb},\bar{\zc})$, if
\begin{gather} \label{eq:eta_zbzc}
\zb_{p,r}=\begin{cases} 
          \za_{p}\cdot \za_r & p \not=r \\
         \binom{\za_p}{2} & p=r
        \end{cases}
        \mbox{\ \ \   and \ \ \  }
\zc_{p,r}=\begin{cases} 
          (\zd_p-\za_{p})\cdot (\zd_r-\za_r) & p \not=r \\
         \binom{\zd_p-\za_p}{2} & p=r
        \end{cases}
\end{gather}
set 
\[
t_{\bar{\za},\bar{\zb},\bar{\zc}}(G)=\sum_{\bar{\zb'},\bar{\zc'}}t_{\bar{\za},\bar{\zb'},\bar{\zc'}}(H)
\]
where the summation is over all valid tuples $\bar{\zb}'=(\zb_{i,j}'\,:\, i,j\in[k])$
and $\bar{\zc}'=(\zc_{i,j}'\,:\, i,j\in[k])$
such that $\zb_{i,j}'=\zb_{i,j}$ and $\zc_{i,j}'=\zc_{i,j}$ if $\{i,j\}\not=\{p,r\}$.
\item For every valid $(\bar{\za},\bar{\zb},\bar{\zc})$, if Equation (\ref{eq:eta_zbzc}) does not hold,
set $t_{\bar{\za},\bar{\zb},\bar{\zc}}(G)=0$.
\end{enumerate}
\item If $(G,\bar{c})=\rho_{p\to r}(H,\bar{c}_H)$: 
\begin{enumerate}
\item For every valid $(\bar{\za},\bar{\zb},\bar{\zc})$ if $\za_{p}=0$,
set 
\[
t_{\bar{\za},\bar{\zb},\bar{\zc}}(G)=\sum_{\bar{\za}',\bar{\zb}',\bar{\zc}'}t_{\bar{\za}',\bar{\zb}',\bar{\zc}'}(H)
\]
where the summation is over all valid tuples $\bar{\za}'=(\za_{i}'\,:\, i\in[k])$,
$\bar{\zb}'=(\zb_{i,j}'\,:\, i,j\in[k])$ and $\bar{\zc}'=(\zc_{i,j}'\,:\, i,j\in[k])$ such that
\begin{itemize}
 \item  $\za_{r}=\za_{p}'+\za_{r}'$,
 \item  $\za_{i}=\za_{i}'$ for all $i\notin \{p,r\}$, 
 \item  for all $j\in [k]\setminus \{p\}$, 
 \[
   \zb_{j,r}=\begin{cases}
   \zb_{j,p}'+\zb_{j,r}' & \mbox{ if }j\not=r\\
   \zb_{r,r}'+\zb_{p,r}'+\zb_{p,p}' & \mbox{ if }j=r\end{cases}
   \mbox{\ \ \  and \ \ \  }
   \zc_{j,r}=\begin{cases}
   \zc_{j,p}'+\zc_{j,r}' & \mbox{ if }j\not=r\\
   \zc_{r,r}'+\zc_{p,r}'+\zc_{p,p}' & \mbox{ if }j=r\end{cases}
\]
and
\item for all $i,j\in[k]\setminus\{p,r\}$, $\zb_{i,j}=\zb_{i,j}'$ and $\zc_{i,j}=\zc_{i,j}'$. 
\end{itemize}
\item For every For every valid $(\bar{\za},\bar{\zb},\bar{\zc})$ if $\za_{p}\not=0$,
set $t_{\bar{\za},\bar{\zb},\bar{\zc}}(G)=0$.
\end{enumerate}
\end{enumerate}
\end{algorithm}
\end{samepage}
\begin{description}
\item [{Correctness}]~
\begin{enumerate}
\item By direct computation.
\item Proved in \cite{ar:AndrenMarkstrom2009} for $Z(G;\t,\y)$. The trivariate case is similar. 
\item $G=\eta_{p,r}(H)$: Let $S$ be a subset of vertices of $V(G)=V(H)$
with $\za_{p}$ and $\za_{r}$ vertices of colors $p$ and $r$ respectively. After
adding all possible edges between vertices of color $p$ and of color
$r$ in $S$, the number of edges between such vertices in $E_G(S)$ is 
$\za_{p}\cdot \za_r$ if $r\not=p$ and $\binom{\za_p}{2}$ if $p=r$. Similarly, 
the number of edges between vertices colored $p$ and $r$ in $E_G(\bar{S})$ is 
$(\zd_p-\za_{p})\cdot (\zd_r -\za_r)$ if $r\not=p$ and $\binom{\zd_p-\za_p}{2}$ if $p=r$. 
\item $G=\rho_{p\to r}(H)$: Let $S$ be a subset of vertices of $V(G)=V(H)$.
After recoloring every vertex
of color $p$ in $S$ to color $r$, we have $\za_{p}=0$. 
Every edge between a vertex colored $p$ to any other vertex lies after the recoloring between
a vertex colored $r$ and another vertex. 
There is one special case, which is the edges that lie between vertices colored $r$ after the recoloring.
Before the recoloring these edges were incident to vertices colored any combination of $p$ and $r$. 
\end{enumerate}

\item [{Running Time}] The size of the $(2^{3cw(G)+2}-1)$-expression is
bounded by $n^{c}\cdot f_{1}(k)$ for some constant $c$, which does
not depend on $cw(G)$, and for some function $f_{1}$ of $cw(G)$. Now
we look at the possible operations performed by Algorithm \ref{alg:cliquewidth}:

\begin{enumerate}
\item The time does not depend on $n$ since $G$ is of size $O(1)$. 
\item The time can be bounded by the size of the table $t_{\bar{\za},\bar{\zb},\bar{\zc}}$
to the power of $3$, i.e. $n^{3 p_{1}(cw(G))}$. 
\item For $\mu_{p,r}$, the algorithm loops over all the values in the table
$t_{\bar{\za},\bar{\zb},\bar{\zc}}$, and for each entry possibly compute a sum
over at most $m_G$ elements. Then, the algorithm loops over all the
values again and performs $O(1)$ operations. 
\item For $\rho_{p\to r}$, the algorithm loops over all the values in the
table $t_{\bar{\za},\bar{\zb},\bar{\zc}}$, and for each entry possibly compute
a sum over elements of the table $t_{\bar{\za},\bar{\zb},\bar{\zc}}$. Then, the
algorithm loops over all the values again and performs $O(1)$ operations. 
\end{enumerate}
Hence, Algorithm \ref{alg:cliquewidth} runs in time $O\left(n_G^{f(cw(G))}\right)$ for some function $f$. 
\footnote{Running times of this kind are refered to as {\em Fixed parameter polynomial time} ($\FPPT$) in \cite{ar:MRAG06}, 
where the computation of various graph polynomials of graphs of bounded clique-width is treated.} 

\end{description}

%%%%%%%%%%%%%%%%%%%%%%%%%%%%%%%%%%%%%%%%%%%%%%%%%%%%%%%%%%%%%%%%%%%%%%%%%%%55
%----------------------------------------------------------------------------------

\section{Conclusion and Open Problems}
Applying the reductions used in the proof of  Theorem \ref{th:mainC} to planar graphs gives again planar graphs. 
Combining Theorem \ref{th:mainC} and its proof
with Lemma \ref{le:R}, a hardness result for the trivariate Ising polynomial on planar graphs
analogous to Theorem \ref{th:mainB} follows. However, both Theorem \ref{th:mainB} and the analog for planar graphs
are not dichotomy theorems since each of them leaves an exceptional set of low dimension unresolved. 
Theorem \ref{th:mainB} serves mainly to suggest the existence of a dichotomy theorem for $Z(G;\x,\y,\z)$
on bipartite planar graphs. 

Another open problem
which arises from the paper is
whether $Z(G;\x,\y,\z)$ requires exponential time to compute in general under $\spETH$. One approach to the latter problem
would be to prove that, say, the permanent or the number of maximum cuts require exponential time under $\spETH$ even
when restricted to regular graphs.%, since for regular graphs $Z(G;\x,\y,\z)$ is essentially the same as $Z(G;\t,\y)$. 

%----------------------------------------------------------------------------------

\section*{Acknowledgements}
I am grateful to my  Ph.D. advisor, Prof. J. A. Makowsky, for drawing my attention to the Ising polynomials and for his 
guidance and support. 

\bibliographystyle{plain}

\end{document}